\documentclass[journal]{IEEEtran}
\IEEEoverridecommandlockouts                              
\bstctlcite{IEEEexample:BSTcontrol} 

\usepackage{enumitem,kantlipsum}
\usepackage{lipsum}
\usepackage{mathtools}
\usepackage{cuted}
\usepackage{romannum}
\usepackage[usenames,dvipsnames]{xcolor}
\usepackage{tkz-berge}
\usetikzlibrary{fit,shapes}                                     
\usepackage{tabularx,booktabs}
\usepackage{setspace}

\newcommand{\authorname}{\\ \\ \\$p(\omega)$ \\ $1-p(\omega)$}
\renewcommand{\thesection}{\arabic{section}}

\usepackage{setspace}
\usepackage{calrsfs}
\DeclareMathAlphabet{\pazocal}{OMS}{zplm}{m}{n}
\usepackage{graphicx,dblfloatfix}
\usepackage{graphics} 
\usepackage{epsfig} 
\usepackage{mathptmx}
\usepackage{caption}
\usepackage{subcaption}
\usepackage{times} 
\usepackage{tikz}
\usetikzlibrary{positioning}
\usepackage{amsmath, amssymb}

\usepackage{amsthm}
\usepackage{amsfonts} 
\usepackage{setspace}
\usepackage{scalefnt}
\usepackage{multirow}
\usepackage{textcomp}
\usepackage[english]{babel}
\usepackage{float} 
\usepackage{algorithmicx}
\usepackage[section]{algorithm}
\usepackage{algpascal}
\usepackage{algc}
\usepackage{algcompatible}
\usepackage{algpseudocode}
\let\oldReturn\Return
\renewcommand{\Return}{\State\oldReturn}
\usepackage{mathtools}
\usepackage{bm}
\usepackage{mathptmx}
\usepackage{times} 
\usepackage{mathtools, cuted}
\usepackage{textcomp}
\usepackage[english]{babel}
\usepackage{rotating}
\usepackage{pgfplots}
\pgfplotsset{compat=1.5}
\usepackage{pgfplotstable}
\usepackage{filecontents}

\usepackage{diagbox}
\usepackage{makecell}

\newtheorem{theorem}{Theorem}
\newtheorem{lemma}[theorem]{Lemma}

\newtheorem{defn}[theorem]{Definition}

\newtheorem{rem}[theorem]{Remark}
\newtheorem{prop}[theorem]{Proposition}

\newcommand{\Y}{\pazocal{Y}}
\newcommand{\W}{\pazocal{W}}

\newcommand{\A}{\pazocal{A}}

\renewcommand{\S}{\pazocal{S}}

\renewcommand{\P}{\pazocal{P}}
\newcommand{\T}{\pazocal{T}}
\newcommand{\E}{\pazocal{E}}

\newcommand{\U}{\pazocal{U}}
\newcommand{\C}{\pazocal{C}}
\newcommand{\D}{\pazocal{D}}
\newcommand{\R}{\mathbb{R}}
\newcommand{\bS}{\pazocal{\bar{S}}}
\newcommand{\bs}{{\bar{s}}}

\newcommand{\sA}{{\scriptscriptstyle{A}}}
\newcommand{\sD}{{\scriptscriptstyle{D}}}
\newcommand{\G}{{\pazocal{G}}}
\newcommand{\F}{\pazocal{F}}
\newcommand{\betA}{\beta^{\sA}_{j}}
\newcommand{\betD}{\beta^{\sD}_{j}}
\newcommand{\alphA}{\alpha^{\sA}}
\newcommand{\alphD}{\alpha^{\sD}}

\graphicspath{{figures/}}
\numberwithin{theorem}{section}

\newcommand{\remove}[1]{}
\def \cN{{\cal N}}

\def \*{\star}
\def \10n{\!\!\!\!\!\!\!\!\!\!}

\def \sF {{\scriptscriptstyle F}}
\usepackage[utf8]{inputenc}
         \DeclareMathAlphabet{\mathscr}{U}{BOONDOX-cal}{m}{n}
         \SetMathAlphabet{\mathscr}{bold}{U}{BOONDOX-cal}{b}{n}
         \DeclareMathAlphabet{\mathbscr} {U}{BOONDOX-cal}{b}{n}

\begin{document}
\title{\LARGE \bf A Game Theoretic Approach for Dynamic Information Flow Tracking  to Detect Multi-Stage Advanced Persistent Threats}
\author{Shana Moothedath, Dinuka Sahabandu,  Joey Allen, Andrew Clark,\\ Linda Bushnell, Wenke Lee, and Radha Poovendran
    \thanks{S. Moothedath, D. Sahabandu, L. Bushnell, and R. Poovendran are with the Department of Electrical and Computer Engineering, University of Washington, Seattle, WA 98195, USA. \texttt{\{sm15, sdinuka, lb2, rp3\}@uw.edu}.}
    \thanks{A. Clark is with the Department of Electrical and Computer Engineering, Worcester Polytechnic Institute, Worcester, MA 01609, USA. \texttt{aclark@wpi.edu}.}
    \thanks{J. Allen and W. Lee are with the College of Computing, Georgia Institute of Technology, Atlanta, GA 30332 USA.
    \texttt{jallen309@gatech.edu, wenke@cc.gatech.edu}.}
    }

\maketitle

\begin{abstract}
Advanced Persistent Threats (APTs) infiltrate cyber systems and compromise specifically targeted data and/or resources through a sequence of stealthy attacks consisting of multiple stages. Dynamic information flow tracking has been proposed to detect APTs.   In this paper, we develop a dynamic information flow tracking game for resource-efficient  detection of APTs via multi-stage dynamic games. The game evolves on an information flow graph, whose nodes are processes and objects (e.g. file, network end points) in the system and the edges capture the interaction between different 
processes and objects.  Each stage of the game has pre-specified targets which are characterized by a set of nodes of the graph and the goal of the APT is to evade detection and reach a target node of that stage. The goal of the defender is to maximize the detection probability while minimizing  performance overhead on the system.  
 The resource costs of the players are different and information structure is asymmetric resulting in a {\em nonzero-sum imperfect information game}.   We  first calculate the best responses of  the players and characterize the set of Nash equilibria for single stage attacks. Subsequently, we provide a polynomial-time algorithm to compute a correlated equilibrium for the multi-stage attack case. Finally, we experiment our model and algorithms on real-world nation state attack data obtained from Refinable Attack Investigation system.
\end{abstract}

\IEEEpeerreviewmaketitle
\section{Introduction}
Advanced Persistent Threats (APTs) are long-term stealthy attacks mounted by intelligent and resourceful adversaries with the goal of sabotaging critical infrastructures  and/or exfiltrating critical information. Typically, APTs target companies and organizations that deal with high-value information and intellectual property. APTs monitor the system for long time and  perform tailored attacks that consist of multiple stages. In the first stage of the attack, APTs start with an initial reconnaissance step followed by an initial compromise. Once the attacker establishes a foothold in the system, the attacker tries to elevate the privileges in the subsequent stages and proceed to the target through more internal compromises. The attacker then performs data exfiltration at an ultra-low-rate. 

Detecting APTs is a challenging task as these attacks are stealthy and customized. However, APTs introduce information flows, as data-flow and control- flow commands, while interacting with the system.  Dynamic Information Flow Tracking (DIFT) is a promising detection mechanism against APTs as DIFT detects adversaries in a system by tracking the traces of the information flows introduced in the system \cite{NewSon-05}. 
DIFT taints or tags sensitive information flows across the system as suspicious and tracks the propagation of  tagged flows through the system and generates security analysis referred as {\em traps}, which are based on certain pre-specified security rules, for any unauthorized usage of tagged data \cite{SuhLeeZhaDev-04}. 

Our objective in this paper is to obtain a resource-efficient  analytical model of DIFT to detect multi-stage APTs by an optimal tagging and trapping procedure. There is an inherent trade-off between the effectiveness of DIFT and the resource costs incurred  due to memory overhead for tagging and tracking non-adversarial information flows. Adversarial interaction makes game theory a promising framework to characterize this trade-off and develop an optimal DIFT defense, which is the contribution of this paper. 
Each  stage of the APT attack is a stage in our multi-stage game model which is characterized by a unique set of critical locations and critical infrastructures of the system, referred to as {\em destinations}.  Note that, the intermediate stages in the attack hold critical information to the adversary for achieving its goals in the final stage.

\noindent The contributions of this paper are the following:
\begin{enumerate}[leftmargin=*]
\item[$\bullet$] We model the  interaction of APTs and DIFT with the system as a  two-player multi-stage nonzero-sum game with imperfect information structure.  The adversary strategizes in each stage of the game to reach a destination node of that specific stage and the defender strategizes  to detect the APT in a resource-efficient manner.   A solution to this game gives an optimal  policy for DIFT that performs selective tagging that minimize both overtagging and undertagging, a tag propagation rule with tag sanitization, and an optimal selection of security rules and trap locations to conduct security analysis to maximize the  probability of APT detection while minimizing memory and performance overhead on the system. 
\item[$\bullet$] We provide algorithms to compute  best responses  of  the adversary and the  defender. The best response of the adversary is obtained by reducing it to a {\em shortest path problem} on a directed graph such that a shortest path gives a sequence of transitions  of the attacker that has  maximum probability of reaching the final target. The best response of the defender, which is a subset of nodes,  is obtained by using the {\em submodularity} property of its  payoff function. 
\item[$\bullet$] We consider a special case of the problem where the attack is a single-stage attack. For this case, we characterize the set of {\em Nash equilibrium} of the game. This characterization is obtained by proving the equivalence of the sequential game to a  suitably defined  bimatrix-game formulation.
\item[$\bullet$] We provide a polynomial-time iterative algorithm to compute a {\em correlated equilibrium} of the game for the multi-stage attack. The correlated equilibrium provide an algorithm to obtain locally optimal equilibrium strategies for both the players  by transforming the two-player game to an $N(M+2)+| \Lambda| + 1$-player game, where $N$ denotes the number of processes and objects in the system, $M$ denotes the number of stages of the APT attack, and $| \Lambda|$ denotes the cardinality of the set of security rules.
\item[$\bullet$] We perform experimental analysis of our model on the real-world multi-stage attack data obtained using Refinable Attack INvestigation (RAIN) framework \cite{JiLeeDowWanFazKimOrsLee-17},  \cite{JiLeeFazAllDowKimOrsLee-18} for a three day nation state attack.
\end{enumerate}

\subsection*{Related Work}\label{sec:rel}
\begin{table*}[t]\caption{An overview of the DIFT architectures for data-flow and control-flow based tracking for different choices of tag sources and security analyses}\label{tb:DIFT}
\begin{center}
\begin{tabularx}{0.98\textwidth}{|p{25mm}|p{5 cm}|p{8.95 cm}|} 
\hline
\diaghead{\theadfont Diag ColumnmnHea}%
{{\bf Reference}}{{\bf DIFT}} &  {\bf Tag source and Tag propagation}   &    {\bf Security analysis} \\
\hline
Newsome et al. \cite{NewSon-05}& \vspace*{-2.5 mm}  \begin{enumerate}[wide, labelwidth=!, labelindent=0pt]
\item[$\bullet$] inputs from  network sockets  
\item[$\bullet$]  data-flow-based  \vspace*{-4.5 mm}  
\end{enumerate} & attacks altering  jump targets,  format string attacks, attacks using  system call arguments, and attacks targeted at specific libraries      \\
\hline
Clause et al. \cite{ClaLiOrs:07}&\vspace*{-2.5 mm}   \begin{enumerate}[wide, labelwidth=!, labelindent=0pt]
\item[$\bullet$] data from  network hosts 
\item[$\bullet$]  data- and control-flow-based    \vspace*{-4.5 mm} 
\end{enumerate} & each instance of call, return, or branch instruction    \\
\hline
Suh et al. \cite{SuhLeeZhaDev-04} &\vspace*{-2.5 mm}  \begin{enumerate}[wide, labelwidth=!, labelindent=0pt]
\item[$\bullet$] all I/O except initial program 
\item[$\bullet$] data-flow-based   \vspace*{-4.5 mm}  
\end{enumerate} &    use of tainted data as load addresses, store addresses, jump targets, and branch conditions \\
\hline
Yin et al. \cite{YinSonEgeChrEng-07} & \vspace*{-2.5 mm}
 \begin{enumerate}[wide, labelwidth=!, labelindent=0pt]
\item[$\bullet$] text, password, HTTP, ICMP, FTP, document, and directory
\item[$\bullet$]  data-flow-based \vspace*{-4.5 mm} 
\end{enumerate}  
  &    anomalous information access, anomalous information leakage, and excessive information access\\
\hline
Vogt et al. \cite{VogNenJovKirKruVig:07} & \vspace*{-2.5 mm}   \begin{enumerate}[wide, labelwidth=!, labelindent=0pt]
\item[$\bullet$] all inputs specified by Netscape 
\item[$\bullet$] data-flow-based \vspace*{-4.5 mm} 
\end{enumerate}    & whenever tainted data is transferred to a third party   \\
\hline
Dalton et al. \cite{DalKanChr-07} & \vspace*{-2.5 mm}   \begin{enumerate}[wide, labelwidth=!, labelindent=0pt]
\item[$\bullet$] every word of memory 
\item[$\bullet$] data-flow-based   \vspace*{-4.5 mm} 
\end{enumerate}&   high level semantic attacks, memory corruption, low-overhead security exceptions \\
\hline
\end{tabularx}
\end{center}
\end{table*}
There are different architectures for DIFT available in the literature to prevent a wide range of attacks  \cite{ClaLiOrs:07}. The fundamental concepts in these architectures remain same, however, they differ in the choice of tagging units, tag propagation rules: data- and control-flow dependencies based rules, and the set of  security rules used for verification of the authenticity of the information flows \cite{SchAvgBru:10}. Table~\ref{tb:DIFT} gives a  brief overview of the different DIFT architectures used in some representative papers. While the papers in Table~\ref{tb:DIFT} gave software modeling of DIFT architecture, we provide an analytical model of DIFT. Specifically, we model DIFT to detect APTs by tracking  information flows using data-flow dependencies.
 
Game theory has been widely used  in the literature to analyse and design  security in cyber systems against different types of adversaries \cite{Tam-11},  \cite{Alp-10}.  For instance, the FlipIt game  modeled in \cite{van2013flipit} captures  the interaction between APTs and the  defender when both the players are trying to take control of a cyber system. In  \cite{van2013flipit},  both APT and defender  take actions periodically and pay a cost for each of their action.  Lee {\em et al.} in \cite{LeeClaAloBusPoo-15} introduced a control-theoretic approach to model competing  malwares in  FlipIt game. 
Game models are available for  APT attacks in  cloud storage \cite{MinXiaXieHajMan-2018} and  cyber systems \cite{RasKonSch-17}. Interaction between an APT and a defender that allocate Central
Processing Units (CPUs) over multiple storage devices in a cloud
storage system is formulated as a Colonel Blotto (zero-sum) game in \cite{MinXiaXieHajMan-2018}. Another zero-sum game model is given in \cite{RasKonSch-17} to model the competition between the APT and the defender in a cyber system.

Often in practice, the resource costs for the defender and the adversary are not  the same, hence the game model is nonzero-sum. In this direction, a nonzero-sum game model is given in \cite{HuLiFuCanMoh-15} to capture the interplay between the defender, the APT attacker,  and the insiders  for joint  attacks.   The approach in \cite{HuLiFuCanMoh-15}  models  the incursion stage of the APT attack,  while our model in this paper captures the different stages of an APT attack. More precisely, we  provide a multi-stage game model that detect APTs by implementing a data-flow-based DIFT detection mechanism while minimizing resource costs. 

A DIFT-based  game model for single-stage attack  is given  in the recent work \cite{SahXiaClaLeePoo-18}. Later, \cite{MooSahClaLeePoo-18} extended the model in \cite{SahXiaClaLeePoo-18}  to  the case of multi-stage attack. The approaches in \cite{SahXiaClaLeePoo-18} and \cite{MooSahClaLeePoo-18} consider a DIFT architecture in which the locations in the system to perform security analysis, called as {\em traps} or tag sinks, are pre-specified and the defender will select the data channels that are to be  tagged. In this paper, we provide an analytical model for data-flow based DIFT architecture that select not only the data channels to be tagged but also the locations to conduct security analysis and also the security rules that are to be verified. The proposed model, hence captures a general model of data-flow based DIFT.
\subsection*{Organization of the Paper}
The rest of the paper is organized as follows:  Section~\ref{sec:pre} describes the preliminaries of DIFT and the system. Section~\ref{sec:formulation} introduces the notations used in the paper and  then presents the game formulation. Section~\ref{sec:equilibria} discusses the solution concept  for the game we consider. Section~\ref{sec:results1} presents a solution approach to the game  for the single-stage attack. Section~\ref{sec:results2} presents a  solution to the game  for the multi-stage attack.  Section~\ref{sec:sim} explains the experimental results of the model and results on real-world data. Finally, Section~\ref{sec:conclu} gives the concluding remarks.

\section{Preliminaries}\label{sec:pre}
In this section, we discuss the detection mechanism DIFT and the graphical representation of the system referred to as {\em information flow graph}.

\subsection{Dynamic Information Flow Tracking}\label{sec:detection_system}
DIFT detection system has three major components: 1)~tag sources, 2)~tag propagation rules, and 3)~tag sinks or traps. 
Tag is a single or multiple bit marking, depending on the level of granularity,  that denotes the sensitivity of a data flow. Data channels, such as keyboards, network interface, and hard disks, are considered as sensitive and hence tagged  by DIFT when it holds information that could be exploited by an APT \cite{VogNenJovKirKruVig:07}.  All information flows emanating from a tagged channel are  tagged flows. 
The tag status of the information flows propagate through the system based on the  pre-specified propagation rules which are either data-flow-based or data- and control-flow-based. Hence, whenever a tagged flow mixes with a benign flow, the resulting flow gets tagged \cite{ClaLiOrs:07}.


Tagged flows are inspected at specific locations called tag sinks also referred as {\em traps} in order to determine the runtime behavior of the system. Tag sinks are specified either using the memory and code locations (like tag sources) or using types of instructions where the users want to analyze a tagged flow before executing certain types of instructions  \cite{ClaLiOrs:07}. Tag sinks are generated in the system when an unusual usage of a tagged information  is detected. The system then obtains the details of the associated flow,  like terminal points of the flow, the path traversed, and concludes if the flow is spurious or not based on the system's or program's security rules.  In case if the system concludes that the flow is spurious, it terminates the system operation. On the other hand, if the flow is  found to be not spurious, then the system continues its operation.  

Conventional DIFT will tag all the sensitive  channels in the system.   
This, however,  results in tagging of numerous authentic flows referred as {\em overtagging} \cite{SchAvgBru:10}   which leads to  false alarms and  performance overhead resulting in system slowdown. On the other hand,  untagged spurious flows due to  {\em undertagging} are security threats to the system. Moreover, conventional DIFT only adds tag and never removes tag leading to {\em tag spread} \cite{SchAvgBru:10}.  
To reduce tag spread and the overhead caused by tagging,  the notion of {\em tag sanitization}  was introduced  in \cite{SchAvgBru:10}. The output of constant operations (where the output is independent of the source data) and a tagged flow successfully passing all security rules can be untagged.   An efficient tagging policy must incorporate tag sanitization and perform selective  tagging  in such a way that both  overtagging and undertagging are minimized.  Also, the selection of security rules and the locations of the  tag sinks  must be optimal to reduce performance and memory  overhead on the system. 
  \subsection{Information Flow Graph}
  Information flow graph $\G= (V_{\G}, E_{\G})$ is a graphical representation of the system in which the node set $V_{\G} =  \{s_1,\ldots, s_N\}$ corresponds to the processes, objects, and files in the system and  edge set $E_{\G} \subseteq V_{\G} \times V_{\G}$ represents  interactions between different nodes. More precisely, the edges of the graph represent  information flows captured using  system log data of the system, for the whole-system execution and workflow during the entire period of logging.  The node set $\D \subset V_{\G}$ denote the  subset of nodes that correspond to critical data centers and the  critical infrastructure sites of the system known as destinations.   We consider multi-stage attacks consisting of, say 
 $M$ stages, where each stage is characterized by a unique set of  destinations.
 The set $\D_j:=\{d_1^j, \ldots, d_{n_j}^j \}$ denotes  the set of destinations in the $j^{\rm th}$ stage of the attack and hence $\D := \cup_{j=1}^M \D_j$. The interaction of DIFT and APTs, which we formally model in Section~\ref{sec:formulation}, evolves through $\G$.
\section{Problem Formulation: Game Model}\label{sec:formulation}
In this section, we model a two player multi-stage game between APTs and DIFT. We model the different stages of the game in such a way that each stage of the APT attack translates to a stage in the game.

\subsection{System Model} 
We denote the adversarial player of the game by $\P_{\sA}$  and the defender player by  $\P_{\sD}$. In the $j^{\rm th}$ stage of the attack, the objective of $\P_{\sA}$ is to evade detection and reach a destination node in stage $j$, given by $\D_j$. The objective of $\P_{\sD}$ is to detect $\P_{\sA}$  before $\P_{\sA}$  reaches a node in $\D_j$. In order to detect $\P_{\sA}$,
$\P_{\sD}$ identifies a set of processes   $\Y := \{y_1,\ldots,y_{h} \} \subseteq V_{\G}$  as the tag sources such that any information flow passing through a process $y_i \in \Y$ is  marked as sensitive.  
$\P_{\sD}$ tracks the  traversal of a tagged  flow through the system  and   generates security analysis  at tag sinks denoted as $\T := \{t_1,\ldots,t_{h'} \} \subset V_{\G}$ using pre-specified rules. 

Let  $\Lambda$ be the set of  security rules.
We consider  security policy that are based on the terminal points of the flow. Therefore, $\Lambda: V_{\G}\times V_{\G} \rightarrow \{0, 1\}$, where $1$ represents that the pair of terminal points of the flow violate the security policy of the system and $0$ otherwise. Here,  $|\Lambda| \leqslant N^2$, since not all node pairs in $V_{\G}$ have a directed path between  them. Hence the number of security rules that are relevant to a node is atmost $N$. Without loss of generality, we assume that each node in $\G$ is associated with $N$ security rules.  As $N$ is large,  applying all $N$ security rules at every tag sink is not often required. In our game model,  DIFT selects a subset of rules  at every tag sink to perform security analysis.  

\subsection{State Space of the Game}
Let $\lambda \subset V_{\G}$ denote the subset of nodes in the information flow graph that are susceptible (vulnerable) to attacks.  In order to characterize the entry point of the attack by a unique node, we introduce a  {\em pseudo-process}  $s_0$ such that $s_0$  is connected to all the processes in the set $\lambda$. 
Let $\mathcal{S} := V_\G \cup \{s_0\}$,  $E_{\lambda} := \{s_0\} \times \lambda$, and $\E := E_{\G} \cup  E_{\lambda}$.
Note that, $s_0$ is the root node  of the modified graph  and hence transitions are allowed {\em from} $s_0$ and no transition is allowed {\em into} $s_0$. 

Now we define the state space of the game. Each decision point in the game is a state of the state space and is defined by the source of the flow in set $\lambda$, the stage of the  attack, the current process $s_i$ along with its tag status, trap status,  and the status of the $N$ security rules applicable at $s_i$.  We use $s_i^j$ to denote the process $s_i$ at the $j^{\rm th}$ stage of the attack. Then the state space of the game is denoted by $\bS  := \{V_{\G} \times \lambda \times \{1, \ldots, M\} \times \{0, 1\}^{2+N}\} \cup \{(s_0^1, \underbrace{0, \ldots, 0}_{2+N\mbox{~times}})\}$, where $\bS= \{\bs_1, \ldots, \bs_{T} \}$ with $T=(2^{(2+N)}NM|\lambda|)+1$. Here  $\bs_1 =(s_0^1, 0, \ldots, 0)$ is the state in $\bS$ corresponding to the pseudo-node $s_0$. The remaining states are given by $\bs_i = (s_i^j, \lambda_i,  k_i^1, \ldots, k_i^{(2+N)})$, for $i=2,\ldots, T$, where  $s_i \in V_{\G}$, $j \in \{1, \ldots, M\}$, $\lambda_i \in \{1, \ldots, |\lambda|\}$,  and $k_i^1, \ldots, k_i^{2+N} \in \{0,1\}$.  Here, $k_i^1=1$ if $s_i$ is tagged and $k_i^1=0$ otherwise. Similarly, $k_i^2=1$ if $s_i$ is a tag sink and $k_i^2 = 0$ otherwise, and $k_i^3, \ldots, k_i^{2+N}$  denotes the selection of security rules (bit $1$ denotes that a rule is selected and bit $0$ denotes that the rule is not selected). Note that $\bS$ has exponential cardinality. Tagging $s_0$ means tagging all sensitive flows which is not desirable on account of the performance overhead. Therefore,  $s_0$ is neither a tag source nor a tag sink and it is always in stage~1 with origin at $s_0$ itself as denoted by state $\bs_1$.  
We give the following definition for an adversarial flow in the state space $\bS$ originating at the state  $(s_0^1, 0, \ldots, 0)$. 
\begin{defn}\label{def:stage-constraint}
An information flow in  $\bS$ that originates at  state $(s^1_0 , 0, \ldots, 0)$ and terminates at  state $(s^{j}_{i}, \lambda_i,  k_i^1, \ldots, k_i^{2+N})$   is said to satisfy the \underline {stage-constraint}  if the flow passes through some destinations  in $\D_1, \D_{2}, \ldots,\D_{ j-1}$ in order.
\end{defn}
\subsection{Actions of the Players}
The players $\P_{\sA}$ and $\P_{\sD}$ have finite action sets over the state space $\bS$ denoted by sets $\A_{\sA}$ and $\A_{\sD}$, respectively.  The action set of $\P_{\sA}$ is  a subset of  $V_{\G}$ and represents the next node in $\G$ that is reached by the flow.  $\P_{\sA}$ can also end the game  by dropping the information flow at any point of time by  transitioning to a {\em null} state $\phi$. 
Thus $\A_{\sA} =\{s_i^j: s_i \in \mathcal{S}, j \in \{1, \ldots, M \}\}\cup \{\phi\}$. Note that, $\lambda_i$ for a state $(s_i^j, \lambda_i,  k_i^1, \ldots, k_i^{(2+N)})$ in $\bS$ is decided by the process in $\lambda$ to which the adversary transitions from $s_0$, i.e., the transition from $(s_0, 0, \ldots, 0)$ in the state space. Further, $\lambda_i$ for a particular adversarial flow remains fixed for all states in $\bS$ that the flow traverses.   As the tag propagation rules are pre-specified by the user, the action set of $\P_{\sD}$ includes selection of tag sources, tag sinks, security check rules, and tag sanitization.  Hence the action set of defender at  $s^j_i$  is a binary tuple, $(k_i^1, \ldots, k_i^{2+N})$, and  $\A_{\sD} = \{0, 1\}^{NM(2+N)}$.  While the objective of $\P_{\sA}$ is to exploit the vulnerable processes $\lambda$  of the system to successfully launch an attack, the objective of $\P_{\sD}$ is to select an optimal set of tagged nodes, say $\Y^\* \subset \Y$, and an optimal set of tag sinks, say $\T^\* \subset \T$,  and a set of security rules  such that any spurious information flow in the system is detected at some tag sink before reaching the destination.

\subsection{Information of the Game}
Both the adversary and the defender know the graph $\G$. At any state $\bs_i$ in the game, the defender has the information about the  tag source status of $\bs_i$,  the tag sink status of $\bs_i$, and the set of security rules chosen at $\bs_i$.  However, the adversary is unaware of  the tag source status, the tag sink status, and the  security rules chosen at that state. 
On the other hand, while the adversary knows the  stage of the attack,  the defender does not know the  stage of the attack and hence the unique set of destinations targeted by $\P_{\sA}$ in that particular stage. Thus,  the players  $\P_{\sA}$ and $\P_{\sD}$  have asymmetric knowledge resulting in an {\em imperfect information} game.

\subsection{Strategies of the Players}
Now we define the strategies of both the players.  A strategy is a rule that the player uses to  select actions at every step of the game. 
Since the action sets of the players are lower level processes with memory constraints and computational limitations, we consider  stationary strategies which are defined below for both the players.
\begin{defn}\label{def:stationary}
A player strategy is stationary if it depends only on the current state.
\end{defn}

Additionally, we consider {\em mixed strategies} and hence there are probability distributions over  the action sets $\A_{\sA}$ and $\A_{\sD}$.   The defender strategy at a process $s_i$, ${\bf p}_{\sD}(s_i)$ is a tuple of length $2+N$, $({{\bf p}^1_{\sD}}(s_i), \ldots, {{\bf p}^{2+N}_{\sD}}(s_i))$, that consists of the probability that $s_i$ is tagged ${{\bf p}^1_{\sD}}(s_i)$, the probability that $s_i$ is a tag sink ${{\bf p}^2_{\sD}}(s_i)$,  and the probability of selecting each rule in $\Lambda$ corresponding to $s_i$, $({{\bf p}^3_{\sD}}(s_i), \ldots, {{\bf p}^{2+N}_{\sD}}(s_i))$. The pseudo-process $s_0$ has ${{\bf p}^{i'}_{\sD}}(s_0) = 0$ for $i' \in \{1, \ldots, 2+N\}$.   Note that the defender strategy does not depend on the stage, as  the defender is unaware of the stage of the attack. 
The adversary on the other hand knows the stage of the attack and hence the strategy of $\P_{\sA}$, i.e., the transition probability distribution ${{\bf p}_{\sA}}: \mathcal{S} \times \{1, \ldots, M\} \rightarrow [0, 1]^{\A_{\sA}}$, depends on the attack stage. Consider a process $s_i$ and let $\cN(s_i)$ denotes the set of neighbors of $s_i$ defined as  $\cN(s_i) := \{s_{i'}: (s_i, s_{i'}) \in \E\} \cup  \{\phi\}$. Then, ${{\bf p}_{\sA}}(s_i^j,  s^{j'}_{i'}) \neq 0$ implies that one of the following cases hold: 1)~$j = j'$ and $ s_{i'} \in \cN(s_i)$, and~2)~$j' = j+1$ and $s_i = s_{i'} \in \D_j$. Here, case~1) corresponds to transition  in the same stage to a neighbor node or dropping out of the game and  case~2) corresponds to transition at a destination from one stage to the next stage. Note that, in case~2) (i.e., $j' = j+1$ and $s_i = s_{i'} \in \D_j$) ${{\bf p}_{\sA}}(s_i^j,  s^{j+1}_{i}) = 1$.   Also,  $\sum_{\substack{s_{i'} \in \cN(s_i)}}{{\bf p}_{\sA}}(s^j_i, s^{j'}_{i'})=1$.
Taken together,  the strategies of $\P_{\sA}$ and $\P_{\sD}$ are given by the vectors ${{\bf p}_{\sD}}=\{({{\bf p}^1_{\sD}}(s_i), \ldots, {{\bf p}^{2+N}_{\sD}}(s_i)): s_i \in \mathcal{S} \}$ and ${{\bf p}_{\sA}}=\{{{\bf p}_{\sA}}(s^j_i, s^{j'}_{i'}): s_i \in \mathcal{S}, j, j' \in \{1, \ldots, M\}, \mbox{~and~} s_{i'} \in \cN(s_i) \}$, respectively. Note that, ${\mathbf{p}}_{\sA}$ is a vector whose length equals the number of edges in the state space $\bS$, say $\hat{|\E|}$, while ${\mathbf{p}}_{\sD}$ is a vector of length $|\mathcal{S}|$ with each entry of length $2+N$.   Notice that ${\bf p}_{\sA}$ is defined in such a way that  a flow that originate at $(s_0^1, 0, \ldots, 0)$ in the state space $\bS$  reaches a state $(s^j_i, \lambda_i,  k_i^1, \ldots, k_i^{2+N})$, for some $\lambda_i \in \{1, \ldots, |\lambda|\}$ and for some $k_i^1, \ldots, k_i^{2+N} \in \{0, 1\}$,  after passing through some destinations of stages~$1, \ldots, j-1$. By this definition of state space and strategies of the game, all  information flows  in $\bS$  satisfy the stage-constraints, given in Definition~\ref{def:stage-constraint}, and can affect the performance of the system and even result in system breakdown, if  malicious.

\subsection{Payoffs to the Players}
Now we define the payoff functions of the players $\P_{\sA}$ and $\P_{\sD}$, denoted by $U_{\sA}$ and $U_{\sD}$, respectively. The payoff function for both the players include penalties and rewards  at every stage of the attack. If the adversarial flow reaches a destination in the $j^{\rm th}$ stage satisfying the stage-constraint, then the adversary earns an {\em intermediate reward} and the defender incurs an {\em intermediate penalty}. On the other hand, if the adversary gets detected at some stage~$j$, then the adversary incurs a {\em penalty}, the defender receives a {\em reward}, and the game terminates. 
 In addition to this, the defender is also associated with  costs for tagging the nodes,  setting tag sinks at the nodes, and selecting  security rules from the set $\Lambda$, as tagging and security analysis of information flows lead to resource overhead such as memory and storage.
   
More precisely, $U_{\sA}$ consists of: (i)~reward $\betA > 0$ for successfully reaching a destination in the $j^{\rm th}$ stage satisfying the stage-constraints, and (ii)~cost $\alphA < 0$ if the adversary is detected by the defender.  Similarly,  $U_{\sD}$ consists of: (a)~memory cost $\C_{\sD}(s_i) < 0$ for tagging node $s_i \in V_{\G}$, (b)~memory cost $\W_{\sD}(s_i) < 0$ for setting tag sink at  node $s_i \in V_{\G}$, (c)~cost $\gamma_i$, for $i \in \{1, \ldots, N\}$, for selecting the $i^{\rm th}$ security check rule at a tag sink, (d)~cost $\betD < 0$ if the adversary reaches a destination in the $j^{\rm th}$ stage satisfying the stage-constraint, and (e)~reward $\alphD > 0$ for detecting the adversary. We assume that the cost of tagging a node and the cost of setting tag sink at a node, $\C_{\sD}(s_i)$ and $\W_{\sD}(s_i)$, respectively,  are independent of the attack stage. However, $\C_{\sD}(s_i)$ and $\W_{\sD}(s_i)$ depends on the average traffic at process $s_i$ and hence $\C_{\sD}(s_i) := c_1\, B(s_i)$ and $\W_{\sD}(s_i) := c_2\, B(s_i)$. Here, $c_1 \in \R_-$ is a fixed tagging cost and $c_2 \in \R_-$ is a fixed cost for setting tag sink, where $\R_-$ is the set of negative real numbers, and $B(s_i)$ denotes the average    traffic at node $s_i$.  

Recall that, the origin of any adversarial information flow in the state space $\bS$ is  $(s_0^1, 0, \ldots, 0)$. 
  For a flow originating at  state $(s_0^1, 0, \ldots, 0)$ in $\bS$,  let $p_{T} (j)$ denotes the probability that  the flow will get detected at stage $j$ and $p_{R}(j)$ denotes the probability that the flow will  reach some destination  in set $\D_{j}$. Note that $p_{T} (j)$ and $p_{R} (j)$ depends on the tag source status,  the tag sink status and also the set of security rules selected.  For a given strategy, ${\bf p}_{\sD}$ and $ {\bf p}_{\sA}$,  the payoffs  $U_{\sD}$ and $U_{\sA}$ are given by,  
\vspace*{-2 mm}
\begin{eqnarray}
U_{\sD}({\bf p}_{\sD}, {\bf p}_{\sA})  &=&  \sum_{s_i \in \S}\Big( {\bf p}^1_{\sD}(s_i)\, \C_{\sD}(s_i) + {\bf p}^2_{\sD}(s_i)\, \W_{\sD}(s_i)+  \nonumber \\
&&\hspace*{-14 mm} \sum_{r=1}^{N} {\bf p}^{2+r}_{\sD}(s_i)\, \gamma_r \Big)+ \sum_{j =1}^M \Big(p_{T} (j)\alpha^{\sD}+ p_{R}(j)  \beta^{\sD}_{j}\Big), \label{eq:Ud}\\
U_{\sA} ({\bf p}_{\sD}, {\bf p}_{\sA}) &=& \sum_{j =1}^M \Big(p_{T} (j)\alpha^{\sA} + p_{R}(j)  \beta^{\sA}_{j}\Big)\label{eq:Ua}.
\end{eqnarray}
\subsection{Preliminary Analysis of the Model}
In this subsection, we perform an initial analysis of our model.
A multi-stage attack consisting of $M$ stages belongs to one of the following $M+2$ scenarios. 
\begin{enumerate}[leftmargin=*]
\item[1)] The adversary drops out of the game before reaching some destination in $\D_1$.
\item[2)] The adversary reaches some destination each in $\D_1, \ldots, \D_{j}$  and then drops out of the game, for $j=1, \ldots, M-1$ ($M-1$ possibilities).
\item[3)] The adversary reaches some destination each in $\D_1, \ldots, \D_{M}$.
\item[4)] The defender detects the adversary at  some stage.
\end{enumerate}
The utility of the game is different for each of the cases listed above. In scenario~1), $\P_{\sA}$ and $\P_{\sD}$  incurs zero payoff. In scenario~2), adversary earns  rewards for reaching stages~$1, \ldots, j$, respectively,   defender incurs penalty for not detecting the adversary at stages $1, \ldots, j$, respectively, and the game terminates. In scenario~3),  the adversary earns  rewards for reaching destinations in all stages and wins the game and the defender incurs a total penalty for not detecting the adversary at all the stages.    In the last scenario,  adversary incurs the penalty for getting detected and the defender earns  the reward for detecting the adversary and wins the game.

\begin{figure*}[!h]
\normalsize
\begin{eqnarray}
\U^{\sD}(s_i^{j'}, \lambda_i,  \bar{k}_i) &=& 
\sum_{s_b \in \S} \Big( p^1_{F, b}(s_i^{j'}, \lambda_i,  \bar{k}_i)\,\C_{\sD}(s_b) + p^2_{F, b}(s_i^{j'}, \lambda_i,  \bar{k}_i)\,\W_{\sD}(s_b) +    \sum_{r=1}^{N} p^{2+r}_{F, b}(s_i^{j'}, \lambda_i, \bar{k}_i)\,\gamma_r \Big)+\nonumber  \\ 
& &   \sum_{j=1}^M\Big( p_{R, j}(s_i^{j'}, \lambda_i, \bar{k}_i) ( \sum_{v=1}^{j}\beta^{\sD}_{v})   + P_{T}(s_i^{j'},\lambda_i,  \bar{k}_i)\alpha^{\sD}\Big), \label{eq:u1} \\
\U^{\sA}(s_i^{j'}, \lambda_i,  \bar{k}_i)&=& 
\sum_{j=1}^M \Big(p_{R, j}(s_i^{j'}, \lambda_i, \bar{k}_i) ( \sum_{v=1}^j  \beta^{\sA}_{v}) +  P_{T}(s_i^{j'}, \lambda_i,  \bar{k}_i)\alpha^{\sA}\Big).\label{eq:u2}
\end{eqnarray}
\hrulefill
\end{figure*}

For calculating the payoffs of $\P_{\sD}$ and $\P_{\sA}$ at a decision point in the game (i.e., at a state in $\bS$), we define  {\em utility functions} $\U^{\sA}: \bS \rightarrow \mathbb{R}$ and $\U^{\sD}: \bS \rightarrow \R$ for the adversary and defender, respectively, at every state in the state space $\bS$. 
Let $q(s_i^{j'})$ denotes the probability with which the adversary  drops out of the game at state $(s_i^{j'}, \lambda_i, k_i^1, \ldots,  k_i^{2+N})$, for any $\lambda_i \in \lambda$ and $k_i^1, \ldots, k_i^{2+N} \in \{0, 1\}$. Let $P_{R,j}(s_{i}^{j^{\prime}}, \lambda_i,  k_i^1, \ldots, k_i^{2+N})$ denotes the probability that an information flow originating at $(s_0^1, 0, \ldots, 0)$ reaches a destination in $\D_{j}$ and then drops out before reaching a destination in $\D_{j+1}$, without getting detected by the defender, when the current state is $(s_{i}^{j^{\prime}}, \lambda_i,  k_i^1, \ldots, k_i^{2+N})$.   Also let $P_{T}(s_{i}^{j'},k_i^1, \ldots, k_i^{2+N})$ denote the probability that an information flow  is detected by the defender when the current state is $(s_i^{j'}, \lambda_i,  k_i^1, \ldots, k_i^{2+N})$. 
To characterize the utility of the players at a state in $\bS$, we now introduce few notations. For notational brevity, let us  denote  $ k_i^1, \ldots, k_i^{2+N}$  by $\bar{k}_i$, for $i=1, \ldots, N$.   For state $(s_{i}^{j^{\prime}}, \lambda_i,  \bar{k}_i)$, define \begin{eqnarray*}{\scalebox{0.85}{\mbox{$Q_r(s_{i}^{j^{\prime}}) \hspace*{-1 mm} := 
\hspace*{-4 mm} \sum\limits_{\substack{s_{\ell} \in \cN(s_i) \\{k_{\ell}^g \in \{0, 1\}}}} \hspace*{-2 mm} {{\bf p}_{\sA}(s_{i}^{j^{\prime}},s_{\ell}^{r})}  \Big[ \prod\limits_{g=1}^{2+N}  \Big({\bf p}^g_{\sD}(s_{\ell})\Big)^{k_{\ell}^g}{\Big(1-\bf p}^g_{\sD}(s_{\ell})\Big)^{(1-k_{\ell}^g)}   \Big]  P_{R,j}(s_{\ell}^{r}, \lambda_i, \bar{k}_{\ell}),$}}}
\end{eqnarray*} 
\begin{eqnarray*}{\scalebox{0.85}{\mbox{$\overline{Q}_r(s_{i}^{j^{\prime}}) \hspace*{-1 mm} :=\hspace*{-4 mm}  \sum\limits_{\substack{s_{\ell} \in \cN(s_i) \\{k_{\ell}^g \in \{0, 1\}}}}\hspace*{-2 mm}   
{{\bf p}_{\sA}(s_{i}^{j^{\prime}},s_{\ell}^{r})}  \Big[ \prod\limits_{g=1}^{2+N}  \Big({\bf p}^g_{\sD}(s_{\ell})\Big)^{k_{\ell}^g}{\Big(1-\bf p}^g_{\sD}(s_{\ell})\Big)^{(1-k_{\ell}^g)}   \Big]  P_{T}(s_{\ell}^{r}, \lambda_i, \bar{k}_{\ell})$}}}.
\end{eqnarray*} 
Then, 
\begin{eqnarray*}
P_{R,j}(s_{i}^{j^{\prime}}, \lambda_i, \bar{ k}_i) \hspace*{-2 mm}&=&\hspace*{-2 mm}
 \left\{
\begin{array}{ll}
0, &  k_i^1=\cdots = k_i^{2+N}=1\\
q(s_{i}^{j^{\prime}}) +  Q_{j'+1}(s_{i}^{j^{\prime}}),  &  s_{i} \in \D_j,~j^{\prime} = j
\\
0, &  s_{i} \in \D_{j^{\prime}}, j^{\prime} = j+1\\
0,  &  j^{\prime} > j+1 \\
 Q_{j'}(s_{i}^{j^{\prime}}), &  j^{\prime} \leq j \\
q(s_{i}^{j^{\prime}}) + Q_{j'}(s_{i}^{j^{\prime}}), 
&  j^{\prime} = j+1\\
\end{array}
\right.\\
P_{T}(s_{i}^{j^{\prime}}, \lambda_i, \bar{k}_i) \hspace*{-2 mm}&=&\hspace*{-2 mm}
\left\{
\begin{array}{ll}
1, & \hspace*{15 mm} k_i^1=\cdots = k_i^{2+N}=1 \\
0, & \hspace*{15 mm}  j'=M,   s_{i} \in \D_{M} \\
\overline{Q}_{j'}(s_{i}^{j^{\prime}}), & \mbox{otherwise.} \\
\end{array}
\right.
\end{eqnarray*}

Using the definitions of $P_{R, j}(\cdot)$ and $P_{T}(\cdot)$ at a state in $\bS$, the payoffs of the defender and the adversary at a state $(s_i^{j'}, \lambda_i,  k_i^1, \ldots,  k_i^{2+N})$ is given by Eqs.~\eqref{eq:u1} and~\eqref{eq:u2} respectively.

In Eqs.~\eqref{eq:u1} and~\eqref{eq:u2}, $p^1_{F, b} (s_i^{j'}, \lambda_i, \bar{k}_i)$ denotes  the probability that node $s_b \in V_{\G}$ is tagged  in a flow whose current state is  $(s_i^{j'}, \lambda_i, \bar{k}_i)$ and $p^2_{F, b} (s_i^{j'},\lambda_i,  \bar{k}_i)$ denotes  the probability that node $s_b \in V_{\G}$ is a tag sink  in a flow whose current state is  $(s_i^{j'}, \lambda_i, \bar{k}_i)$. Similarly, $p^{2+r}_{F, b} (s_i^{j'}, \lambda_i, \bar{k}_i)$ denotes  the probability that the $r^{\rm th}$ security rule is selected for inspecting authenticity of a flow whose current state is $(s_i^{j'}, \lambda_i, \bar{k}_i)$.  Eqs.~\eqref{eq:u1} and~\eqref{eq:u2}   give a system of $2^{(2+N)}NM|\lambda|+1$ linear equations each for the utility vectors  $\U^{\sD}$ and $\U^{\sA}$, where $\U_{\sA}(b)$, $\U_{\sD}(b)$ denote the utilities at the $b^{\rm th}$ state in $\bS$. 
Now we give the following result, which  relates global payoffs  $U_{\sD}, U_{\sA}$ with local payoffs $\U_{\sD}, \U_{\sA}$, respectively.

\begin{lemma}\label{lem:s_0}
Consider the defender and adversary strategies ${\bf p}_{\sD}$ and ${\bf p}_{\sA}$, respectively. Then, the following hold: (i)~$U_{\sA}({\bf p}_{\sD}, {\bf p}_{\sA}) = \U_{\sA}(s_0^1, 0, \ldots, 0)$,
and (ii)~$U_{\sD}({\bf p}_{\sD}, {\bf p}_{\sA}) = \U_{\sD}(s_0^1, 0, \ldots, 0)$.
\end{lemma}

\begin{proof}
\noindent{\bf (i)}: By definition, 
$\U^{\sA}(s_0^1, 0, \ldots, 0) = 
\sum_{j=1}^M \Big(P_{R, j}(s_0^1, 0, \ldots, 0)  ( \sum_{v=1}^j  \beta^{\sA}_{v} )+  P_{T}(s_0^1, 0, \ldots, 0)\alpha^{\sA}\Big)$.  Here,
\begin{multline}\label{eq:ua}
\sum_{j=1}^M P_{R, j}(s_0^1, 0, \ldots, 0)  ( \sum_{v=1}^j  \beta^{\sA}_{v} )= \beta^{\sA}_1 \, \sum_{j=1}^M P_{R, j}(s_0^1, 0, \ldots, 0)  + \\
 \beta^{\sA}_2 \, \sum_{j=2}^M P_{R, j}(s_0^1, 0, \ldots, 0) + \ldots + \beta^{\sA}_M \, P_{R, M}(s_0^1, 0, \ldots, 0),
\end{multline}
Where, $\sum_{j=1}^M p_{R, j}(s_0^1, 0, \ldots, 0)$ is the total probability that a flow originating at $(s_0^1, 0, \ldots, 0)$ reach some destination in $\D_1$. Similarly, $\sum_{j=2}^M p_{R, j}(s_0^1, 0, \ldots, 0)$ is the total probability that a flow originating at $(s_0^1, 0, \ldots, 0)$ reach some destination in $\D_2$. Thus
\begin{multline}\label{eq:uaaaa}
\sum_{j=1}^M P_{R, j}(s_0^1, 0, \ldots, 0) = p_R(1), \sum_{j=2}^M P_{R, j}(s_0^1, 0, \ldots, 0) = p_R(2), \ldots, \\
P_{R, M}(s_0^1, 0, \ldots, 0) = p_R(M).
\end{multline}
From Eqs.~\eqref{eq:ua} and~\eqref{eq:uaaaa}, we get
\begin{equation}\label{eq:uaaa}
\sum_{j=1}^M \Big( P_{R, j}(s_0^1, 0, \ldots, 0)   (\sum_{v=1}^j  \beta^{\sA}_{v})\Big)  =  \sum_{j=1}^M p_R(j) \beta^{\sA}_j.
\end{equation}
Since  $P_{T}(s_0^1, 0, \ldots, 0)\ = \sum_{j=1}^M p_T(j)$,
\begin{equation}
P_{T}(s_0^1, 0, \ldots, 0)\alpha^{\sA} = \sum_{j=1}^M p_{T}(j) \alpha^{\sA}.  \label{eq:uaa}
\end{equation}
From Eqs.~\eqref{eq:uaaa} and~\eqref{eq:uaa}, we get $\U^{\sA}(s_0^1, 0, \ldots, 0) = \sum_{j=1}^M \Big(p_R(j) \beta^{\sA}_j +p_T(j) \alpha^{\sA}\Big) = U_{\sA}({\bf p}_{\sD}, {\bf p}_{\sA})$. 

\noindent{\bf (ii)}: Notice that $p^1_{F, i}(s_0^1, 0, \ldots, 0)$ is the probability that the process $s_i$ is a tag source in a flow originating at $(s_0^1, 0, \ldots, 0)$. Thus $p^1_{F, i}(s_0^1, 0, \ldots, 0) = {\bf p}^1_{\sD}(s_i)$. Similarly, we get $p^2_{F, i}(s_0^1, 0, \ldots, 0) = {\bf p}^2_{\sD}(s_i)$ and $p^{2+r}_{F, i}(s_0^1, 0, \ldots, 0) = {\bf p}^{2+r}_{\sD}(s_i)$ for $r=1, \ldots, N$. This along with Eqs.~\eqref{eq:uaaa} and~\eqref{eq:uaa} implies that  $\U^{\sD}(s_0^1, 0, \ldots, 0) = \sum_{s_i \in \S} \Big( {\bf p}^1_{\sD}(s_i) \C_{\sD}(s_i) +{\bf p}^2_{\sD}(s_i) \W_{\sD}(s_i) + \sum_{r=1}^{N} {\bf p}^{2+r}_{\sD}(s_i) \gamma_r\Big) + \sum_{j=1}^M \Big(p_R(j) \beta^{\sD}_j +p_T(j) \alpha^{\sD}\Big) = U_{\sD}({\bf p}_{\sD}, {\bf p}_{\sA})$. This completes the proof of (i)~and~(ii).
\end{proof}

\section{Game Model: Solution Concept}\label{sec:equilibria}
This section presents an overview of the notions of equilibrium considered in this work. We first describe the concept of a player's best response to a given mixed policy of an opponent.

\begin{defn}
\label{def:BR}
Let ${\mathbf{p}}_{\sA}: \{\S \times \{1,\ldots,M\}\}\cup \{s_0^1\} \rightarrow [0,1]^{|\hat{\E}|}$ denote an adversary strategy (transition probabilities) and ${\mathbf{p}}_{\sD}: \S \rightarrow [0,1]^{(2+N)|\mathcal{S}|}$ denote a defender strategy (probabilities of tagging, tag sink selection, and security rule selection at every node in the  graph). The set of best responses of the defender  given by $$\mbox{BR}({\mathbf{p}}_{\sA}) = \arg\max_{{\mathbf{p}}_{\sD}}{\{U^{\sD}({\mathbf{p}}_{\sD}, {\mathbf{p}}_{\sA}) :{\mathbf{p}}_{\sD} \in [0,1]^{(2+N)|\mathcal{S}|}\}}.$$ Similarly, the best responses of the adversary are given by $$\mbox{BR}({\mathbf{p}}_{\sD}) = \arg\max_{{\mathbf{p}}_{\sA}}{\{U^{\sA}({\mathbf{p}}_{\sD}, {\mathbf{p}}_{\sA}) : {\mathbf{p}}_{\sA} \in [0,1]^{|\hat{\E}|}\}}.$$
\end{defn}
Intuitively, the best responses of the defender are the set of tagging strategies, the set of tag sink selection strategies, and the set of security rule selection strategies that jointly maximize the defender's utility for a given adversary strategy. At the same time, the best responses of the adversary are the sets of transition probabilities that maximize the adversary's utility for a given defender (tagging, tag sink selection, and security rule selection) strategy. A mixed policy  profile is a {\em Nash equilibrium} (NE) if the mixed policy  of each player is a best response to the fixed mixed policy  of the rest of the players.  Formal definition of Nash equilibrium is as follows.
\begin{defn}
\label{def:NE}
A pair of mixed policies $({\mathbf{p}}_{\sD}, {\mathbf{p}}_{\sA})$ is a \emph{Nash equilibrium} if 
$${\mathbf{p}}_{\sD} \in \mbox{BR}({\mathbf{p}}_{\sA}) \mbox{~and}  \quad{\mathbf{p}}_{\sA} \in \mbox{BR}({\mathbf{p}}_{\sD}).$$
\end{defn}
A Nash equilibrium captures the notion of a stable solution as it occurs when neither player can improve its payoff by unilaterally changing its strategy. Unilateral deviation of  the adversary's strategy is  a change in one of the transition probabilities for fixed defender's strategy and unilateral deviation of  the defender's strategy is  a change in either tagging probability, or tag sink selection probability, or the probability of selecting a security rule at a node, for fixed adversary strategy.  Kuhn's  equivalence result \cite{Kuh-53} between mixed and stochastic policies along with Nash's result in  \cite{Nash-50} that prove the existence of  a Nash equilibrium (NE) for a finite game with mixed strategy, guarantees the existence of NE for the game we consider in this paper.
While there exists a Nash equilibrium for games with rational, noncooperative players, it is NP-hard to compute it in general, especially for nonzero-sum dynamic games of the type considered in this paper.  Also note that, for the game considered in this paper, the utility functions for the  players  are nonlinear in the probabilities. A weaker solution concept which is a relaxation of the Nash equilibrium is the \emph{correlated equilibrium}  defined as follows.

\begin{defn}
\label{def:correlated}
Let $P$ denote a joint probability distribution over the set of defender and adversary actions. The distribution $P$ is a \emph{correlated equilibrium} if for all strategies ${\mathbf{p}}_{\sA}^{\prime}$ and ${\mathbf{p}}_{\sD}^{\prime}$, 
\begin{eqnarray*}
\displaystyle \mathbf{E}_{({\mathbf{p}}_{\sD}, {\mathbf{p}}_{\sA} )\sim P}(U^{\sD}({\mathbf{p}}_{\sD}, {\mathbf{p}}_{\sA})) &\geq& \mathbf{E}_{({\mathbf{p}}^{\prime}_{\sD}, {\mathbf{p}}_{\sA}) \sim P}(U^{\sD}({\mathbf{p}}_{\sD}^{\prime}, {\mathbf{p}}_{\sA}) \\
\mathbf{E}_{({\mathbf{p}}_{\sD}, {\mathbf{p}}_{\sA}) \sim P}(U^{\sA}({\mathbf{p}}_{\sD}, {\mathbf{p}}_{\sA})) &\geq& 
\mathbf{E}_{({\mathbf{p}}_{\sD}, {\mathbf{p}}^{\prime}_{\sA}) \sim P}(U^{\sA}({\mathbf{p}}_{\sD}, {\mathbf{p}}_{\sA}^{\prime})
\end{eqnarray*}
\end{defn}

Here, $\mathbf{E}(\cdot)$ denotes the expectation.  We next consider a simpler version of the correlated equilibrium that models the local policies at each process. 

\begin{defn}
\label{def:local_correlated}
Let $P$ denote a joint probability distribution over the set of defender and adversary actions.  The distribution $P$ is a \emph{local correlated equilibrium} if for all states $s_{i} \in \mathcal{S}$, $j \in \{1,\ldots,M\}$, and strategies $p_{\sD}^{\prime}(s_{i})$ and $p_{\sA}^{\prime}(s_{i}^{j}, \cdot)$, we have  
\begin{eqnarray*}
\mathbf{E}_{({\mathbf{p}}_{\sD}, {\mathbf{p}}_{\sA}) \sim P}(U^{\sD}({\mathbf{p}}_{\sD}, {\mathbf{p}}_{\sA})) &\geq& \mathbf{E}_{({\mathbf{p}}^{\prime}_{\sD}, {\mathbf{p}}_{\sA}) \sim P}(U^{\sD}({\mathbf{p}}_{\sD}^{\prime}, {\mathbf{p}}_{\sA}) \\
\mathbf{E}_{({\mathbf{p}}_{\sD}, {\mathbf{p}}_{\sA}) \sim P}(U^{\sA}({\mathbf{p}}_{\sD}, {\mathbf{p}}_{\sA})) &\geq& 
\mathbf{E}_{({\mathbf{p}}_{\sD}, {\mathbf{p}}^{\prime}_{\sA}) \sim P}(U^{\sA}({\mathbf{p}}_{\sD}, {\mathbf{p}}_{\sA}^{\prime})
\end{eqnarray*}
where $\mathbf{p}_{\sD}^{\prime}$ denotes a strategy with ${\bf p'}_{\sD}^{x}(s_{i}) = {p'}_{\sD}^{x}(s_{i})$, for some $x \in \{1, \ldots, 2+N\}$,  ${\bf p'}_{\sD}^{y}(s_i) = {\bf p}_{\sD}^{y}(s_i)$ for $y\in \{1, \ldots, 2+N\}, y \neq x$, and ${\bf p}_{\sD}^{\prime}(s_{i^{\prime}}) = {\bf p}_{\sD}(s_{i^{\prime}})$ for $i \neq i^{\prime}$, and $\mathbf{p}_{\sA}^{\prime}$ denotes a strategy with $\mathbf{p}_{\sA}^{\prime}(s_{i}^{j},\cdot) = p_{\sA}^{\prime}(s_{i}^{j},\cdot)$ and $\mathbf{p}_{\sA}^{\prime}(s_{i^{\prime}}^{j^{\prime}},\cdot) = {\bf p}_{\sA}(s_{i^{\prime}}^{j^{\prime}}, \cdot)$ for $(i,j) \neq (i^{\prime},j^{\prime})$.
\end{defn}
 \section{Best Response of the Players}\label{sec/;BR}
 In this section, we calculate the best responses of both the players, $\P_{\sA}$ and $\P_{\sD}$.
 \subsection{Best Response for the Adversary}\label{subsec:BR_a}
The best response of the adversary to a given defender strategy is described here. Firstly, we present 
the following preliminary lemma.
\begin{lemma}
\label{lemma:adversary_BR}
Consider a defender  policy ${\bf p}_{\sD}$. For each destination $d_{b}^{j} \in \D_{j}$, let $\Omega_{d_{b}^{j}}$ denote the set of paths in $\bS$ that originate at $(s_0,0, \ldots, 0)$ and terminate at some state that correspond to node $d_{b}^{j}$.  For any path $\omega$, let $p(\omega)$ denote the probability that a flow reaches the destination without getting detected by the adversary. Finally, for every $d_{b}^{j}$, choose a path $\omega_{d_{b}^{j}}^{\ast} \in \arg\max{\{p(\omega) : \omega \in \Omega_{d_{b}^{j}}\}}$. Let $\omega^{\ast} \in \arg\max{\{p(\omega_{d_{b}^{j}}) : d_{b}^{j} \in \D_{j}, j=1,\ldots,M\}}.$ Finally, define the policy ${\bf p}_{\sA}^{\ast}$ by
\begin{displaymath}
{\bf p}_{\sA}^{\ast}(s_{i}^{j},s_{i^{\prime}}^{j^{\prime}}) = \left\{
\begin{array}{ll}
 1, & (s_{i}^{j}, s_{i^{\prime}}^{j^{\prime}}) \in \omega^{\ast} \\
 0, & \mbox{else}
\end{array}
\right.
\end{displaymath}
Then,  $\omega^{\ast} \in \mbox{BR}({\bf p}_{\sD})$.
\end{lemma}

\begin{proof}
Let ${\bf p}_{\sA}$ be any adversary policy, and let $\Omega$ denote the set of paths that are chosen by the policy with nonzero probability such that the termination of the path is at some destination in $\D = \cup_{j=1}^M \D_j$. The utility of the adversary can be written as 
\begin{eqnarray*}
U^{\sA} &=& \sum_{\omega \in \Omega}{\pi(\omega)(p(\omega)\beta_{j(\omega)}^{\sA} + (1-p(\omega))\alpha^{\sA})} \\
&=& \sum_{j=1}^{M}{\sum_{d_{b}^{j} \in \D_{j}}{\sum_{\omega \in \Omega_{d_{b}^{j}}}{\pi(\omega)(p(\omega)\beta^{\sA}_{j} + (1-p(\omega))\alpha^{\sA}}}},
\end{eqnarray*}
 where $j(\omega)$ is equal to the stage where the path terminates and $\pi(\omega)$ is the probability that the path is chosen under this policy. The utility $U^{\sA}$ is then bounded above by the path that maximizes $p(\omega)(\beta^{\sA}_{j}-\alpha^{\sA})$, which is exactly the path $\omega^{\ast}$. 
\end{proof}

Using Lemma \ref{lemma:adversary_BR}, we present the following approach to select a best response to the adversary for a given defender policy. For each destination in $\bigcup_{j=1}^{M}{\D_{j}}$, we first choose a path $\omega$ to that destination such that the probability of reaching that destination, $p(\omega)$, is maximized while traversing destinations of all intermediate stages. From those paths, we then select  a path that maximizes $p(\omega)(\beta^{\sA}_{j}-\alpha^{\sA})$. 

\begin{prop}\label{prop:A_BR_opt}
The path $\omega^{\ast}$ returned by a shortest path algorithm on the state space   graph with edge weights of each incoming edge to states that correspond to node $s_i$ equal to $\Big( {\bf p}^1_{\sD}(s_i){\bf p}^2_{\sD}(s_i)\prod_{r=3}^{2+N}{\bf  p}^r_{\sD}(s_{i})\Big)$ is a best response to the defender strategy ${\bf p}_{\sD}$.
\end{prop}
\begin{proof}
Consider a path $\omega \in \Omega_{d_{b}^{j}}$, i.e., a path that originate at $(s_0, 0, \ldots, 0)$ and terminate at some state that correspond to node $d_b^j$. Let the first  node, that belongs to the vulnerable set $\lambda$,  through which  $\omega$ traverse be denoted by $\lambda_{\omega}$. 
Then for all nodes of $\G$ that lie in $\omega$, i.e.,  $s_i \in \omega$, let $\Lambda^{\omega}(s_i)$ denote the set of indices of the security rules in $\Lambda$ that are based on  $s_i$ and the pair $(\lambda_{\omega}, s_i)$.

By Lemma \ref{lemma:adversary_BR}, it suffices to show that a shortest path in the state space with a suitably defined weight function will return  a path with maximum probability of reaching some state in $\bS$  corresponding to node $d_{b}^{j}$ of $\G$ without getting detected by the adversary. For any path $\omega \in \Omega_{d_{b}^{j}}$, the probability that the flow reaches $d_{b}^{j}$ without getting detected by the adversary is equal to $\prod_{s_{i} \in \omega}{\Big(1- {\bf p}^1_{\sD}(s_i){\bf p}^2_{\sD}(s_i)\prod_{r=3}^{2+N}{\bf  p}^r_{\sD}(s_{i})\Big)}.$
\begin{eqnarray*}
&&\max \prod_{s_{i} \in \omega}{\Big(1- {\bf p}^1_{\sD}(s_i){\bf p}^2_{\sD}(s_i)\prod_{r=3}^{2+N}{\bf  p}^r_{\sD}(s_{i})\Big)} \\
&&= \max \sum_{s_{i} \in \omega}\log\, {\Big(1-{\bf p}^1_{\sD}(s_i){\bf p}^2_{\sD}(s_i)\prod_{r=3}^{2+N}{\bf  p}^r_{\sD}(s_{i})\Big)} \\
&&= \min \sum_{s_{i} \in \omega} {\Big( {\bf p}^1_{\sD}(s_i){\bf p}^2_{\sD}(s_i)\prod_{r=3}^{2+N}{\bf  p}^r_{\sD}(s_{i})\Big)}
\end{eqnarray*}

The problem of finding best response to the adversary is equivalent to finding the shortest path from $s_{0}$ to $d_{b}^{j}$ in a graph where the edge weights are equal to $\Big( {\bf p}^1_{\sD}(s_i){\bf p}^2_{\sD}(s_i)\prod_{r=3}^{2+N}{\bf  p}^r_{\sD}(s_{i})\Big)$ for each edge incoming to $s_{i}^{j^{\prime}}$, for $j' \in \{1, \ldots, M\}$. 
\end{proof}
 \subsection{Best Response for the Defender}\label{subsec:BR_d}
We now present an approach for approximating the best response of the defender. In this approach, the set of possible responses at $s_{i}$ is discretized. Define $$V_r = \{s_{i}^{z_r} : s_{i}  \in \S, z_r = 1,\ldots,Z_r\}$$ for integers $Z_r > 0$, where $r=1, \ldots, 2+N$. For any $V_r^{\prime} \subseteq V_r$, $r=1, \ldots, 2+N$, define $p_{\sD}(s_{i} ; V_r^{\prime}) = \frac{1}{Z_r}|\{s_{i}^{z_r} : z_r=1,\ldots,Z_r\} \cap V_r^{\prime}|$, and define $p_{\sD}(V_1^{\prime})$ to be the resulting vector of probabilities for tag source selection, $p_{\sD}(V_2^{\prime})$ to be the resulting vector of probabilities for tag sink selection, and $p_{\sD}(V_3^{\prime}), \ldots, p_{\sD}(V_{2+N}^{\prime})$ to be the resulting vectors of probabilities for security rule selection.  Then $p_{\sD}(V^{\prime}) = \{p_{\sD}(V_r^{\prime})\}_{r=1}^{2+N} \}$ is the resulting vector of defender strategy, where $V' = \{V'_1, \ldots, V'_{2+N} \}$.  For a given adversary strategy, say $p_{\sA}$, let $f(V^{\prime}) = U^{\sD}(p_{\sD}(V^{\prime}), {p}_{\sA})$.


\begin{prop}
\label{prop:BR_submod}
The function $f(V^{\prime})$ is submodular as a function of $V^{\prime}$, that is, for any $V_r^{\prime}$, $V_r^{\prime\prime}$ with $V_r^{\prime} \subseteq V_r^{\prime\prime}$ and any $s_{i}^{z_r} \notin V_r^{\prime\prime}$, for  $r\in \{1, \ldots, 2+N\}$, $$f\Big(\Big\{V_r^{\prime} \cup \{s_{i}^{z_r}\}\Big\}_{r=1}^{2+N}\Big) - f(V^{\prime}) \geq f\Big(\Big\{V_r^{''} \cup \{s_{i}^{z_r}\}\Big\}_{r=1}^{2+N}\Big) - f(V^{\prime\prime}).$$
\end{prop}

\begin{proof}
Consider $U^{\sD}$ as defined in Eq. (\ref{eq:Ud}).  The first and second terms of $U^{\sD}$ are equal to $$\sum_{s_{i} \in \S}{\frac{\C^{\sD}(s_{i})}{Z_1}|\{s_{i}^{z_1} : z_1 = 1,\ldots,Z_1\} \cap V_1^{\prime}|} \mbox{~and~}$$ $$ \sum_{s_{i} \in \S}{\frac{\W^{\sD}(s_{i})}{Z_2}|\{s_{i}^{z_2} : z_2 = 1,\ldots,Z_2\} \cap V_2^{\prime}|},$$ respectively, both of which are modular as a function of $V^{\prime}$. The third term of $U^{\sD}$ equals $$\sum_{s_{i} \in \S}\, \sum_{r=1}^{N}{\frac{\gamma_r}{Z_{2+r}}|\{s_{i}^{z_{2+r}} : z_{2+r} = 1,\ldots,Z_{2+r}\} \cap V_{2+r}^{\prime}|},$$ which  is also modular as a function of $V^{\prime}$.
The last term can be written as $$\sum_{\omega}{\pi(\omega)\sum_{j=1}^{M}{(p_{T}(j;\omega)\alpha^{\sD} + p_{R}(j;\omega)\beta_{j}^{\sD})}},$$ where $p_{T}(j;\omega)$ (resp. $p_{R}(j;\omega)$) denotes the probability that the adversarial flow is detected by the defender at the $j^{\rm th}$ stage (resp. reaches some destination in stage~$j$) when the sample path is $\omega$ and the defender strategy is $p_{\sD}(V^{\prime})$ (the $V^{\prime}$ is omitted from the notation for simplicity). $\pi(\omega)$ denotes the probability of selecting the path $\omega$. Since the last destination  that is reached before dropping out is determined by the choice of path (denote this destination $j(\omega)$), we have 
\begin{eqnarray*}
\sum_{j=1}^{M}\hspace*{-1 mm}{(p_{T}(j;\omega)\alpha^{\sD} \hspace*{-1 mm}+\hspace*{-1 mm} p_{R}(j;\omega)\beta_{j}^{\sD})}\hspace*{-3 mm} &=& \hspace*{-3 mm} g(\omega ; V^{\prime})\alpha^{\sD} \hspace*{-1 mm}+ \hspace*{-1 mm} (1\hspace*{-0.5 mm}-\hspace*{-0.5 mm}g(\omega ; V^{\prime}))\beta^{\sD}_{j(\omega)}\\
& =&\hspace*{-2.8 mm} g(\omega ; V^{\prime})(\alpha^{\sD}-\beta_{j(\omega)}^{\sD}) + \beta^{\sD}_{j(\omega)}. \end{eqnarray*}
Since $\alpha^{\sD}-\beta_{j(\omega)}^{\sD} \geq 0$ and $\beta_{j(\omega)}^{\sD}$ is independent of $p_{\sD}(V')$, it suffices to show that $g(\omega ; V^{\prime})$ is submodular as a function of $V^{\prime}$. Let $V^{\prime} \subseteq V^{\prime\prime}$  with $V_r^{\prime} \subseteq V_r^{\prime\prime}$  and $s_{i}^{z_r} \notin V_r^{\prime\prime}$, for any $r\in \{1, \ldots, 2+N\}$. We can write $g(\omega ; V^{\prime\prime})$
\begin{eqnarray*}
 &=& 1 - \left[\prod_{\stackrel{s_{i_{k}} \in \omega:}{i_{k} = i}}{(1-\prod_{r=1}^{2+N}p^r_{\sD}(s_{i_{k}}))}\right]\left[\prod_{\stackrel{s_{i_{k}} \in \omega:}{i_{k} \neq i}}{(1-\prod_{r=1}^{2+N}p^r_{\sD}(s_{i_{k}}))}\right] \\
&=& 1 - \delta(V^{\prime\prime})(1-\prod_{r=1}^{2+N}p^r_{\sD}(s_{i}))^{c(s_{i} ; \omega)},
\end{eqnarray*}
where $ p^r_{\sD}(s_{i_k})$, for $r=1, \ldots, 2+N$, denotes the probabilities of selecting node $s_{i_k}$ as tag source, as tag sink, and selecting security rules under the policy $p_{\sD}(V')$, respectively, and  $$\delta(V^{\prime\prime}) = \prod_{\stackrel{s_{i_{k}} \in \omega:}{i_{k} \neq i}}{(1-\prod_{r=1}^{2+N}p^r_{\sD}(s_{i_{k}}))}, \quad c(s_{i} ; \omega) = |\{s_{i_{k}} \in \omega : i_{k} = i\}|.$$ Hence,  $g(\omega ; \Big\{V_r^{\prime\prime} \cup \{s_{i}^{z_r}\}\Big\}_{r=1}^{2+N}) - g(\omega ; V^{\prime\prime})$
\begin{eqnarray*}
& = & \scalebox{1}{\mbox{$1 - \delta(V^{\prime\prime})(1-\prod_{r=1}^{2+N}(p^r_{\sD}(s_{i} ; V^{\prime\prime}) + \frac{1}{Z_r}))^{c(s_{i} ; \omega)} -$}}\\
&& (1-\delta(V^{\prime\prime})(1-\prod_{r=1}^{2+N}p^r_{\sD}(s_{i} ; V^{\prime\prime}))^{c(s_{i};\omega)} )\\
& = & \delta(V^{\prime\prime})\Big[(1-\prod_{r=1}^{2+N}p^r_{\sD}(s_{i} ; V^{\prime\prime}))^{c(s_{i} ; \omega)}-\\
&& (1-\prod_{r=1}^{2+N}(p^r_{\sD}(s_{i} ; V^{\prime\prime}) - \frac{1}{Z_r}))^{c(s_{i} ; \omega)}\Big]
\end{eqnarray*}
When $V_r^{\prime} \subseteq V_r^{\prime\prime}$, $p^r_{\sD}(s_{i} ; V^{\prime}) \leq p^r_{\sD}(s_{i} ; V^{\prime\prime})$, and hence 
\begin{eqnarray*}
(1-\prod_{r=1}^{2+N}p^r_{\sD}(s_{i};V^{\prime\prime}))^{c(s_{i} ; \omega)} - (1-\prod_{r=1}^{2+N}(p^r_{\sD}(s_{i} ; V^{\prime\prime}) - \frac{1}{Z_r}))^{c(s_{i} ; \omega)} \\
\leq (1-\prod_{r=1}^{2+N}p^r_{\sD}(s_{i};V^{\prime}))^{c(s_{i} ; \omega)} - (1-\prod_{r=1}^{2+N}(p^r_{\sD}(s_{i} ; V^{\prime}) - \frac{1}{Z_r}))^{c(s_{i} ; \omega)}.
\end{eqnarray*}
  Furthermore, $V_r^{\prime} \subseteq V_r^{\prime\prime}$ for $r=\{1, \ldots, 2+N\}$ implies $\delta(V^{\prime}) \geq \delta(V^{\prime\prime})$. Hence $$g(\omega ; \Big\{V_r^{\prime} \cup \{s_{i}^{z_r}\}\Big\}_{r=1}^{2+N}) - g(\omega ; V^{\prime})$$ $$ \geqslant g(\omega ; \Big\{V_r^{\prime\prime} \cup \{s_{i}^{z_r}\}\Big\}_{r=1}^{2+N}) - g(\omega ; V^{\prime\prime}),$$ completing the proof of submodularity.
\end{proof}

Submodularity of $f(V^{\prime})$ implies the following.

\begin{prop}
\label{prop:submod_bound}
There exists an algorithm that is guaranteed to select a set $V^{\*}$ satisfying $f(V^{\*}) \geq \frac{1}{2}\max{\{f(V^{\prime}) : V^{\prime} \subseteq V\}}$ within $O(NZ)$ evaluations of $U^{\sD}$.
\end{prop}

\begin{proof}
The proof follows from submodularity of $V^{\prime}$ and \cite{BucFelSefSch-15}.
\end{proof}
\section{Results: Solution to Flow Tracking Game for Single-Stage Attacks}
\label{sec:results1}
In this section, we focus on the  case where there is only a single attack stage and provide a solution for the DIFT game.  Recall that in single-stage attack, the attacker's objective is to choose transitions in the information flow graph so as to reach a target node. Our approach to solve the game is based on a minimum capacity cut-set formulation on a flow network constructed for the information flow graph of the system followed by  solving a bimatrix game. Here, $M=1$ and hence we drop the notation for stage in this section.

\noindent For a flow-network $\F$, a {\em cut} is defined below.
\begin{defn}\label{def:feasible_flow}
In a flow-network $\F$ with vertex and directed edge sets $V_{\sF}$ and $E_{\sF}$ respectively,  for  a subset $\hat{\S} \subset V_{\sF}$ the cut induced by $\hat{\S}$ is a subset of edges $\kappa(\hat{\S}) \subset E_{\sF}$ such that for every $(u,v) \in \kappa(\hat{\S})$, $|\{u,v\} \cap \hat{\S}| = 1$.
\end{defn}
 The set $\kappa(\hat{\S})$ consists of all edges whose one end point is in $\hat{\S}$. Given a flow-network $\F = (V_{\sF}, E_{\sF})$ with source-sink pair $(s_{\sF}, t_{\sF})$ and edge capacity vector $c_{\sF}: E_{\sF} \rightarrow \R_{+}$, the cost of a cut $\kappa(\hat{\S})$, is defined  as the sum of the costs of the edges in the cut
\begin{align} \label{eq:flow}
c_{\sF}(\kappa(\hat{\S})) = \sum_{e \in \kappa(\hat{\S})} c_{\sF}(e).
\end{align}
The objective of the  {\em (source-sink)-min-cut problem} is to find a cut $\kappa(\hat{\S}^\*)$  of  $\hat{\S}^\*$ such that $c_{\sF}(\kappa(\hat{\S}^\*)) \leqslant c_{\sF}(\kappa(\hat{\S}))$
for any cut  $\kappa(\hat{\S})$ of $\hat{\S}$ satisfying $s_{\sF} \in \hat{\S}$ and $t_{\sF} \notin \hat{\S}$. The (source-sink)-min-cut problem is  well studied and there exist different algorithms that find the maximum flow $f^\*$ in  polynomial time (polynomial in $|V_{\sF}|$ and $|E_{\sF}|$)  \cite{Orl:93}. Given an information flow graph $\G$, we first construct the flow-network $\F$.

Pseudo-code describing the construction of $\F=(V_{\sF}, E_{\sF})$ is given in Algorithm~\ref{alg:NE}. The vertex set of $\F$ consists of two nodes $s_i$ and $s'_i$ corresponding to each node $s_i$ in the information flow graph $\G$ and additional vertices $s_{\sF}, t_{\sF}$ (Step~\ref{step:vertex}). Thus $|V_{\sF}|=2N+2$. 
The directed edge set of $\F$ consists of all edges in the information flow graph ($\bar{E}_{\G}$), edges corresponding to the nodes in the information flow graph ($\hat{E}_{\G}$), edges connecting source node $s_{\sF}$ to all nodes in the set $\lambda$ ($E_{\lambda}$), and edges connecting all destination nodes to the sink node  $t_{\sF}$ ($E_{\sD}$) (Step~\ref{step:edge}).
The capacity vector $c_{\sF}$ is defined in such a way that all edges except the edges corresponding to nodes in $\G$ have infinite capacity. The capacity of the edges in $\hat{E}_{\G}$ are defined as the sum of  cost for selecting those nodes as tag source and tag sink, since the costs of selecting the security rules do not depend on the node (Step~\ref{step:capa}). Hence, a minimum capacity edge in $\F$ corresponds to a node in $\G$ that has minimum cost of tagging and trapping.  Let  $\kappa(\hat{\S}^\*)$ denotes  an optimal solution to the (source-sink)-min-cut problem on $\F$. Since $\hat{E}_{\G}$ is a cut and $\sum_{e \in \hat{E}_{\G}}c_{\sF}(e) < \infty$,  $\kappa(\hat{\S}^\*) \subset \hat{E}_{\G}$ Then,  the min-cut nodes is given by
\begin{equation}\label{eq:mincut}
 \hat{\S}^\* := \{s_i: (s_i, s'_i) \in \kappa(\hat{\S}^\*)\}.
 \end{equation}
\begin{algorithm}[t]
\setstretch{1.05}
  \caption{Pseudo-code for constructing the flow-network $\F$ and defender payoff function $U_{\sD}(\cdot)$
  \label{alg:NE}}
  \begin{algorithmic}
\State \textit {\bf Input:} Information flow graph $\G$, costs $\C_{\sD}, \W_{\sD}, \gamma_1, \ldots, \gamma_{N}$
\State \textit{\bf Output:} Flow-network $\F$, source, sink nodes: $s_{\sF},t_{\sF}$, capacity vector $c_{\sF}$
\end{algorithmic}
  \begin{algorithmic}[1]
  \State Construct flow-network $\F$ with vertex set $V_{\sF}$ and edge set $E_{\sF}$ as follows:  \label{step:graph} 
  \State  $V_{\sF} \leftarrow V_{\G} \cup V'_{\G} \cup \{s_{\sF},  t_{\sF}\}$, where $V_{\G}=\{s_1, \ldots, s_N\}$, $V'_{\G} = \{s'_1, \ldots, s'_N\}$, and $s_{\sF} = s_0$\label{step:vertex}  
  \State $E_{\sF} \leftarrow \bar{E}_{\G} \cup \hat{E}_{\G} \cup E_{\lambda} \cup E_{\sD}$, where $\bar{E}_{\G} = \{(s'_i, s_j): (s_i, s_j) \in E_{\G}\}$, $\hat{E}_{\G}=\{(s_i, s'_i): i=1,\ldots, N\}$, $E_{\lambda} = \{(s_{\sF}, s_i): s_i \in \lambda$, and $E_{\sD}=\{(s'_i, t_{\sF}): s_i \in \D\}$\label{step:edge}    
\State $c_{\sF} (e)  \leftarrow \begin{cases}
\begin{array}{ll}
\infty, & e \in \bar{E}_{\G} \cup  E_{\lambda} \cup E_{\sD}\\
\C_{\sD}(s_i)+\W_{\sD}(s_i), & e \in \hat{E}_{\G} 
\end{array}
\end{cases}$ \label{step:capa}
\end{algorithmic}
\end{algorithm}
The objective of the defender is to optimally select a defense policy such that no adversarial flow reaches from $s_0$ to some node in $\D$. In other words, defender ensures that no flow from  $s_{\sF}$ reaches $t_{\sF}$ without getting detected. For achieving this the defender's policy must have a nonzero probability of tag and trap for at least  one node in all possible paths from $s_{\sF}$ to $t_{\sF}$. For any adversary policy, the best possible choice for the defender is to tag and trap at a node that has the minimum total cost $\C_{\sD}(\cdot)+ \W_{\sD}(\cdot)$.  An attack path is a directed path  from $s_0$ to some node in $\D$ formed by a sequence of transitions of the adversary. The probability of an attack path under an adversary strategy is the product of the probabilities of all transitions along that path. The adversary plans its transitions to obtain an attack path with least probability of detection. The result below characterizes Nash equilibria of the single-stage.

\begin{theorem}\label{thm:mincutNE}
Let $ \hat{\S}^\*$ be a min-cut of the flow-network $\F=(V_{\sF}, E_{\sF})$ constructed in Algorithm~\ref{alg:NE}. Then, at Nash equilibrium for the single-stage attack case, the  defender's policy is to tag and trap with equal probability all the nodes in $ \hat{\S}^\*$. Further, the adversary's strategy is such that each attack path passes through exactly one node in $ \hat{\S}^\*$. 
\end{theorem}
Before giving the proof of Theorem~\ref{thm:mincutNE}, we present the following lemma that establishes the first main argument in proving the theorem.
\begin{lemma}\label{lem:NEdisj}
Let $\Omega_{\D}$ be the set of all paths in $\G$ from $s_0$ to some node in $\D$ under any adversary policy.  Then, for a defender policy that assign tag and trap at all nodes in the min-cut $ \hat{\S}^\*$ and does not tag and trap any other node, the best response of the adversary is a sequence of transitions such that any attack path (or set of paths if mixed policy)  passes through exactly one node which is tagged and also a trap.
\end{lemma}
\begin{proof}
Consider a  policy of the defender where all the nodes in the min-cut, i.e., $\hat{\S}^\*$,  are tagged and assigned as traps. Note that, all paths in $\Omega_{\D}$ passes through some node in $\hat{\S}^\*$. We prove the argument through a contradiction. Suppose that there exists  a path $ \omega \in \Omega_{\D}$ such that there are two nodes, say $s_i, s_r$,  in path $\omega$ with nonzero probability of tag and trap. Without loss of generality, assume that in $\omega$ there exists a directed path from  $s_i$ to $s_r$.  Now we show that $\pi(\omega) = 0$, where $\pi(\omega)$ is the probability with which adversary chooses path $\omega$. Note that $s_i, s_r \in \hat{\S}^\*$. Since $\hat{\S}^\*$ corresponds to a min-cut, there exists paths in $\Omega_{\D}$ that has node $s_i$ in it but not $s_r$, and vice-versa.  Hence for an adversary whose current state is $s_i$, there exists a path from $s_i$ to some node in $\D$ that guarantees the win of adversary. The transition probability from $s_i$ to a node in $\omega$ that will lead to  some node in $\D$ through $s_r$ is  zero as this path has lower adversary payoff. This gives $\pi(\omega) = 0$  and completes the proof.
\end{proof}

The following result proves  that, for any adversary policy  the best response of the defender is to tag and trap at one node in every  attack path under that adversary policy.
\begin{lemma}\label{lem:NEBR}
Let $\Omega_{\D}$ denote the set of all  paths from $s_0$ to some node in $\D$ under any adversary policy. If the defender's policy is such that the probability of detecting the adversary is the same for all  $\omega \in \Omega_{\D}$, then the best response of the defender is always to tag at exactly one node in every $\omega \in \Omega_{\D}$. 
\end{lemma}
\begin{proof}
Given $(1-p(\omega))$'s are same for all $\omega \in \Omega_{\D}$. Consider a defender's policy ${\bf p}_{\sD}$  in which exactly one node in every $\omega \in \Omega_{\D}$ is chosen as the tag source and tag sink. Assume that the defender policy is modified to ${\bf p}'_{\sD}$ such that more than one node in some path has nonzero tag and trap probability. This variation updates the probabilities of nodes in a set of paths in $\Omega_{\D}$. For ${\bf p}'_{\sD}$  to be  a best response, $U_{\sD}({\bf p}'_{\sD}, {\bf p}_{\sA}) \geqslant U_{\sD}({\bf p}_{\sD}, {\bf p}_{\sA}) $.  The defender's payoff is given by
\begin{multline*}
U_{\sD}({\bf p}_{\sD}, {\bf p}_{\sA}) = \sum_{\omega \in \Omega_{\D}}\pi(\omega)\Big[(1-p(\omega))\alpha^{\sD} + p(\omega)\beta^{\sD} +\\
 \sum_{s_i \in \omega} [{\bf p}^1_{\sD}(s_i)\C_{\sD}(s_i)+  {\bf p}^2_{\sD}(s_i)\W_{\sD}(s_i) + \sum_{r=1}^{N} {\bf p}^{2+r}_{\sD}(s_i) \gamma_{r}]\Big]
\end{multline*}
The terms in $U_{\sD}$ that correspond to $\alpha^{\sD}$ and $\beta^{\sD}$ are the same in both the cases as $p(\omega)$'s are equal for all possible paths. Hence the terms in $U_{\sD}$ differ  in the terms corresponding to $\C_{\sD}$, $\W_{\sD}$, and $\gamma_r$'s.  Note that defender's probabilities (policy) at two nodes in  a path are dependent due to the constraint on $p(\omega)$. Hence for every path  whose probabilities are modified,  $\sum_{s_i \in \omega} [{\bf p}^{'1}_{\sD}(s_i)\C_{\sD}(s_i)+  {\bf p}^{'2}_{\sD}(s_i)\W_{\sD}(s_i) + \sum_{r=1}^{N} {\bf p}^{'r}_{\sD}(s_i) \gamma_{r}] < \sum_{s_i \in \omega} [{\bf p}^1_{\sD}(s_i)\C_{\sD}(s_i)+  {\bf p}^2_{\sD}(s_i)\W_{\sD}(s_i) + \sum_{r=1}^{N} {\bf p}^r_{\sD}(s_i) \gamma_{r}]$, as the probability in the single node case is less than the sum of the probabilities of more than one node case as the events are dependent and the $\C_{\sD}$ and $\W_{\sD}$ values are also the least possible ($\gamma_r$'s are equal at all nodes in the information flow graph). This implies  $U_{\sD}({\bf p}'_{\sD}, {\bf p}_{\sA}) < U_{\sD}({\bf p}_{\sD}, {\bf p}_{\sA})$. Therefore, there exists no best response for the defender which has more than one node with nonzero tag and trap probability in a path, if $p(\omega)$'s are equal for all $\omega \in \Omega_{\D}$. 
\end{proof}

The result below deduces a property of the best response of the adversary which along with Lemma~\ref{lem:NEBR} establishes the final main argument to prove Theorem~\ref{thm:mincutNE}. 
\begin{lemma}\label{lem:detection_equal}
Let $\Omega_{\D}$ denote the set of all paths from $s_0$ to some node in $\D$ under any adversarial policy. Let the defender's policy is such that the probability of detecting the adversary is the same for all  $\omega \in \Omega_{\D}$. Then the best response of the defender is to tag and trap the flows at the min-cut of the flow-network $\F$ constructed in Algorithm~\ref{alg:NE}.
\end{lemma}
\begin{proof}
By Lemma~\ref{lem:NEBR} the best response of the defender is to tag and trap at   one node in every attack path. Note that all attack paths chosen under  ${\bf p}_{\sA}$ passes though some node in the min-cut.   Assigning nonzero probability of tag and trap at the nodes in $\hat{\S}^\*$, all possible attack paths have some nonzero probability of getting detected. We  prove  the result using a contradiction argument. Suppose that it is not the best response of the defender to tag and trap the nodes in $\hat{\S}^\*$. Then, there exists a subset of nodes $\hat{\S} \subset \S$ such that $\sum_{s_i \in \hat{\S}} \C_{\sD}(s_i) + \W_{\sD}(s_i) < \sum_{s_r \in \hat{\S}^\*} \C_{\sD}(s_r) + \W_{\sD}(s_r) $ and all possible paths from $s_0$ to nodes in $\D$ pass through some node in $\hat{\S}$. Then,  $\hat{\S}$ is a (source-sink)-cut-set and let $\kappa(\hat{\S}) := \{(s_i, s'_i): s_i \in \hat{\S}\}$. Then, $\kappa(\hat{\S})$ is a cut set and $c_{\sF}(\kappa(\hat{\S})) < c_{\sF} (\kappa(\hat{\S}^\*))$. This contradicts the fact that $\kappa(\hat{\S}^\*)$ is an optimal solution to the (source-sink)-min-cut problem.  Hence the best response of the defender is to tag and trap only the nodes in $\hat{\S}^\*$. 
\end{proof}

\noindent  Now we present the proof of Theorem~\ref{thm:mincutNE}.\\
{{\bf {\em Proof~of~Theorem~\ref{thm:mincutNE}}}:
Lemma~\ref{lem:NEdisj}  proves that the best response of the adversary is any sequence of transitions that gives a path (or set of paths if mixed policy) that passes through exactly one node that is a tag source and a trap, if the defender's policy is to tag at the min-cut nodes.
Lemma~\ref{lem:detection_equal}  concludes that the best response of the defender is to tag and trap the adversary at the nodes in the min-cut of the flow-network $\F$, provided the probability of detecting the adversary in all $\omega \in \Omega_{\D}$ are equal.  This implies that, if the detection probability ($(1-p(\omega))$'s) are equal at  NE, then  the defender's policy at NE will tag and trap the nodes in the min-cut   and the adversary will choose an attack path such that it passes through exactly one node that is tagged and also a trap.  Now we show that the detection probability are indeed equal at NE.
 
Consider any unilateral deviation in the policy of the adversary. Let $\pi(\omega)$'s for $\omega \in \Omega$ are modified due to change in transition probabilities of the adversary such that the  updated  probabilities of the attack paths are $\pi(\omega_i)+\epsilon_i$, for $i=1, \ldots, |\Omega|$. Here, $\epsilon_i$'s can take positive values, negative values or zero such that $\sum_{i=1}^{|\Omega|}\epsilon_i = 0$.  Consider two arbitrary paths, say $\omega_1$ and $\omega_2$, such that a unilateral change in the adversary policy changes $\pi(\omega_1)$ and $\pi(\omega_2)$ and the probabilities of the other paths remain unchanged. Without loss of generality, assume that $\pi(\omega_1)$ increases by $\epsilon$ while   $\pi(\omega_2)$ decreases by $\epsilon$ and all other $\pi(\omega)$'s remain same. As $({\bf p}_{\sD}, {\bf p}_{\sA})$ is a Nash equilibrium $(\pi(\omega_1) + \epsilon)\Big( p(\omega_1)(\beta^{\sA} - \alpha^{\sA}) + \alpha^{\sA} \Big) + (\pi(\omega_2) - \epsilon)\Big( p(\omega_2)(\beta^{\sA} - \alpha^{\sA}) + \alpha^{\sA} \Big) \leqslant$ $\pi(\omega_1) \Big( p(\omega_1)(\beta^{\sA} - \alpha^{\sA}) + \alpha^{\sA} \Big) + \pi(\omega_2) \Big( p(\omega_2)(\beta^{\sA} - \alpha^{\sA}) + \alpha^{\sA} \Big) $. This implies
$(p(\omega_1)-p(\omega_2))(\beta^{\sA} - \alpha^{\sA}) \leqslant 0$. As $(\beta^{\sA} - \alpha^{\sA}) \geqslant 0$, this implies $(p(\omega_1)-p(\omega_2)) \leqslant 0$. By exchanging the roles of $\omega_1$ and $\omega_2$ and using the same argument one can also show that  $(p(\omega_1)-p(\omega_2)) \geqslant 0$. This implies $(p(\omega_1)-p(\omega_2)) =0$. 
Since $\omega_1$ and $\omega_2$ are arbitrary, one can show that for the general case 
\begin{equation}\label{eq:GlobalUa}
\sum_{i=1}^{|\Omega|}\epsilon_ip(\omega_i) = 0.
\end{equation}
Eq.~\eqref{eq:GlobalUa} should hold for all possible values of $\epsilon_i$'s satisfying 
$\sum_{i=1}^{|\Omega|}\epsilon_i = 0.$ This gives $p(\omega_i) = p(\omega_j)$ for all $i, j \in \{1, \ldots , |\Omega|\}$  at Nash equilibrium. This completes the proof.
\qed

The NE of the game is characterized by a solution of the min-cut problem and the set of transitions of the adversary such that all attack paths have exactly one node which is tagged and also a trap. Moreover, the tag and trap probability of these nodes are equal. Note that, the solution to the min-cut problem is not unique in a general flow graph and hence the NE of the game may not be unique. Results in this subsection conclude that  at any NE the defender's policy will tag and trap the nodes in the min-cut with equal detection probability and the adversary chooses its transitions such that in every attack path exactly one node is a tag source and a tag sink.
\subsection{Matrix Game Analysis for Nash Equilibrium}\label{subsec:NEpolicy}
\begin{table*}[h]\caption{Matrix game for single-stage case with disjoint attack paths}\label{tb:matrix}
\begin{center}
{\scshape
\begin{tabular}{@{}c@{}}
\authorname 
\end{tabular}}
\begin{tabular}{|p{20mm}|c|c|c|c|}
\hline
\diaghead{\theadfont Diag ColumnmnHea}%
{Defender}{Adversary} & $\hat{s}_1$& $\hat{s}_2$ & \ldots & $\hat{s}_a$ \\
\hline
Not detected & \phantom{$9^{9^9}$}$(\beta^{\sD}, \beta^{\sA})$ & $(\beta^{\sD}, \beta^{\sA})$ & \ldots & $(\beta^{\sD}, \beta^{\sA})$ \\
\hline
Detected &  \phantom{$9^{9^9}$}$(\alpha^{\sD} +\mbox{cost}(\hat{s}_1), \alpha^{\sA})$  & $(\alpha^{\sD} +\mbox{cost}(\hat{s}_2), \alpha^{\sA})$& \ldots & $(\alpha^{\sD} +\mbox{cost}(\hat{s}_a), \alpha^{\sA})$\\
\hline
\end{tabular}
\end{center}
\end{table*}
%
In this subsection, we discuss the matrix-game formulation  of the single-stage case give in Table~\ref{tb:matrix}. We first solve the (source-sink)-min-cut problem on $\F$. Let an optimal solution be $\kappa(\hat{\S}^\*)$.  Let the vertex set corresponding to $\kappa(\hat{\S}^\*)$ be $\hat{\S^\*} = \{\hat{s}_{1}, \ldots, \hat{s}_a  \}$, where $\hat{\S^\*} := \{s_i: (s_i, s'_i) \in \kappa(\hat{\S}^\*)\}$.  
By Theorem~\ref{thm:mincutNE}, at NE the defender only tag and trap the nodes in $\hat{\S^\*}$ and adversary chooses transitions such that it passes through only one tag and trap. The  attack paths  chosen by the adversary are therefore characterized by $\{\hat{s}_{1}, \ldots, \hat{s}_a  \}$.
  We denote the probability of selecting an attack path  corresponding to the node $\hat{s}_{i}$ as $\pi(\hat{s}_{i})$. Further, let $ \mbox{cost}(\hat{s}_i)$ denote the total cost of selecting node $\hat{s}_{i}$ for conducting  security analysis.
  
\begin{rem}
At NE, since any attack path pass through exactly one node in $\hat{\S^\*} = \{\hat{s}_{1}, \ldots, \hat{s}_a  \}$ which indeed has equal probability of tag and trap, without loss of generality, one can consider the action space of the adversary as the set of disjoint paths through $\hat{\S^\*}$ and the adversary strategize over this set of disjoint paths. Thus in the bimatrix formulation, the strategy of the adversary is to select a path which  is uniquely defined by a node in $\hat{\S^\*} $.
\end{rem}

Now, we present a result that characterizes the set of NE of the single-stage attack case of the game given in Section~\ref{sec:formulation}.
\begin{theorem}\label{th:NEmatrix}
Solution to the matrix-game given in Table~\ref{tb:matrix} gives Nash equilibrium for the single-stage flow tracking game.
\end{theorem}
\begin{proof}
 The defender's payoffs  for the two cases in Table~\ref{tb:matrix} is 
\begin{eqnarray}
U_{\sD}(\mbox{Not detected}) &=& \sum_{i=1}^{a} \pi(\hat{s}_i) \Big( \beta^{\sD} \Big)\label{eq:NEUd_1} \\
U_{\sD}(\mbox{Detected}) &=& \sum_{i=1}^{a} \pi(\hat{s}_i) \Big( \alpha^{\sD} + \mbox{cost}(\hat{s}_i) \Big)\label{eq:NEUd_2}
\end{eqnarray}
If  $U_{\sD}(\mbox{Not detected}) > U_{\sD}(\mbox{Detected})$, then the defender will never tag and trap at any node and the adversary will reach the destination without getting detected by the defender. On the other hand, if  $U_{\sD}(\mbox{Detected}) > U_{\sD}(\mbox{Not detected})$, then the adversary will always tag and trap all the nodes in $\hat{\S}^\*$ with probability one. The defender will randomly choose between to detect the adversary or not if Eqs.~\eqref{eq:NEUd_1} and~\eqref{eq:NEUd_2}  are equal.  This gives 
\begin{equation}\label{eq:matrixD}
 \sum_{i=1}^{a} \pi(\hat{s}_i) \Big(\beta^{\sD}-\alpha^{\sD} - \mbox{cost}(\hat{s}_i) \Big) =0.
 \end{equation}
 There are many possible values of $\pi(\hat{s}_i)$'s for $i=1, \ldots, a$, that satisfy Eq.~\eqref{eq:matrixD}. Each of those solution will give a probability mixture, i.e., $\pi(\hat{s}_i)$'s, for the adversary at a Nash equilibrium. 
 
 In order to obtain the probability mixture of the defender, we consider the following in Table~\ref{tb:matrix}. For every $\hat{s}_{i} \in \hat{\S^\*} $, one can find the set of nodes that belong to the set $\lambda$ that has directed path to the node $\hat{s}_{i}$ using a depth first search (DFS) algorithm \cite{CorLeiRivSte:01}. Let this set be denoted by $\lambda(\hat{s}_{i})$. Then, cost($\hat{s}_{i}$) $= \C_{\sD}(\hat{s}_{i})+\W_{\sD}(\hat{s}_{i})+\sum_{r \in \lambda(\hat{s}_{i})}\gamma_r$ for all  $\hat{s}_{i} \in \hat{\S^\*} $.  Then the probability of not detecting the adversary in an attack path with min-cut node $\hat{s}_{i}$ in it is given by $(1 - {\bf p}^1_{\sD}(\hat{s}_{i}) {\bf p}^2_{\sD}(\hat{s}_{i})\prod_{r \in \lambda(\hat{s}_{i})} {\bf p}^{2+r}_{\sD}(\hat{s}_{i}))$.
\begin{eqnarray}
 U_{\sA}(\hat{s}_1) &=& \Big(  1- {\bf p}^1_{\sD}(\hat{s}_1)\,{\bf p}^2_{\sD}(\hat{s}_1)\,\prod_{r \in \lambda(\hat{s}_1)} {\bf p}^{2+r}_{\sD}(\hat{s}_1) \Big)\beta^{\sA} + \nonumber \\
 && {\bf p}^1_{\sD}(\hat{s}_1)\,{\bf p}^2_{\sD}(\hat{s}_1)\,\prod_{r \in \lambda(\hat{s}_1)} {\bf p}^{2+r}_{\sD}(\hat{s}_1) \alpha^{\sA}\label{eq:NEUa_1} \\
U_{\sA}(\hat{s}_2) &=& \Big(  1- {\bf p}^1_{\sD}(\hat{s}_2)\,{\bf p}^2_{\sD}(\hat{s}_2)\,\prod_{r \in \lambda(\hat{s}_2)} {\bf p}^{2+r}_{\sD}(\hat{s}_2) \Big)\beta^{\sA} +  \nonumber\\
&& {\bf p}^1_{\sD}(\hat{s}_1)\,{\bf p}^2_{\sD}(\hat{s}_1)\,\prod_{r \in \lambda(\hat{s}_2)} {\bf p}^{2+r}_{\sD}(\hat{s}_2) \alpha^{\sA}\label{eq:NEUa_2}  \\
  \vdots & & \vdots \nonumber\\
U_{\sA}(\hat{s}_a) &=& \Big(  1- {\bf p}^1_{\sD}(\hat{s}_a)\,{\bf p}^2_{\sD}(\hat{s}_a)\,\prod_{r \in \lambda(\hat{s}_a)} {\bf p}^{2+r}_{\sD}(\hat{s}_a) \Big)\beta^{\sA} +  \nonumber\\
&& {\bf p}^1_{\sD}(\hat{s}_1)\,{\bf p}^2_{\sD}(\hat{s}_1)\,\prod_{r \in \lambda(\hat{s}_a)} {\bf p}^{2+r}_{\sD}(\hat{s}_a) \alpha^{\sA}\label{eq:NEUa_n} 
\end{eqnarray}

If the adversary's payoff with respect to one of the node in $\hat{\S}^\*$ is greater than the rest, then the adversary will always select the attack path with respect to that node. The adversary will randomly choose between attack paths that correspond to nodes  $\hat{s}_1, \hat{s}_2,  \ldots, \hat{s}_a$ only when $ U_{\sA}(\hat{s}_1) =  U_{\sA}(\hat{s}_2) = \ldots=  U_{\sA}(\hat{s}_a)$ which means Eqs.~\eqref{eq:NEUa_1} to~\eqref{eq:NEUa_n} must be equal. Theorem~\ref{thm:mincutNE} says that $p(\omega) = (1 - {\bf p}^1_{\sD}(\hat{s}_{1}) {\bf p}^2_{\sD}(\hat{s}_{1})\prod_{r \in \lambda(\hat{s}_{1})} {\bf p}^{2+r}_{\sD}(\hat{s}_{1})) = \ldots = (1 - {\bf p}^1_{\sD}(\hat{s}_{a}) {\bf p}^2_{\sD}(\hat{s}_{a})\prod_{r \in \lambda(\hat{s}_{a})} {\bf p}^{2+r}_{\sD}(\hat{s}_{a}))$. Thus Eqs.~\eqref{eq:NEUa_1} to~\eqref{eq:NEUa_n} are the same.

The defender's payoff is given by $$U_{\sD}({\bf p}_{\sD}, {\bf p}_{\sA}) = \sum_{\omega \in \Omega}\pi(\omega)\Big[(1-p(\omega))\alpha^{\sD} + p(\omega)\beta^{\sD} +$$ $$ \sum_{s_i \in \omega} [{\bf p}^1_{\sD}(s_i)\C_{\sD}(s_i)+  {\bf p}^2_{\sD}(s_i)\W_{\sD}(s_i) + \sum_{r=1}^{N} {\bf p}^{2+r}_{\sD}(s_i) \gamma_{r}]\Big]$$
At Nash equilibrium with respect to the solution $\hat{\S}^\*$ of the (source-sink)-min-cut game, changing  ${\bf p}^{2+r}_{\sD}(\hat{s}_i)$ for any value of $r \in \{1, \ldots, N\}$ and for one node, say $\hat{s}_i \in \hat{\S}^\*$, will not improve the payoff $U_{\sD}$.  Firstly, let us assume that the tagging probability at $\hat{s}_i$ changes from ${\bf p}^1_{\sD}(\hat{s}_i)$ to ${\bf p'}^1_{\sD}(\hat{s}_i)$. By equilibrium condition $U_{\sD}({\bf p}_{\sD}, {\bf p}_{\sA}) \geqslant U_{\sD}({\bf p'}_{\sD}, {\bf p}_{\sA})$. Note that by Lemma~\ref{lem:NEdisj} each node with nonzero defender's probability will lie in exactly one chosen path by the adversary.  Hence $$ \pi(\hat{s}_i ) \Big[p(\omega)(\beta^{\sD} - \alpha^{\sD}) + \alpha^{\sD} + {\bf p}^1_{\sD}(\hat{s}_i ) \C_{\sD}(\hat{s}_i )\Big] $$ $$\geqslant \pi(\hat{s}_i )\Big[p'(\omega)(\beta^{\sD} - \alpha^{\sD}) + \alpha^{\sD} + {\bf p'}^1_{\sD}(\hat{s}_i ) \C_{\sD}(\hat{s}_i )\Big], \mbox{~and~}$$ 
$$ \pi(\hat{s}_i )\Big[(p(\omega) - p'(\omega))(\beta^{\sD} - \alpha^{\sD}) + ( {\bf p}^1_{\sD}(\hat{s}_i ) - {\bf p'}^1_{\sD}(\hat{s}_i )) \C_{\sD}(\hat{s}_i )\Big] \geqslant 0.$$
Here $$p(\omega) - p'(\omega) =\underbrace{ \Big( {\bf p}^2_{\sD}(\hat{s}_i) \prod_{\substack{r \in \lambda(\hat{s}_i)}} {\bf p}^{2+r}_{\sD}(\hat{s}_i) \Big)}_{\varphi_1(\hat{s}_i)}\Big({\bf p'}^1_{\sD}(\hat{s}_i) - {\bf p}^1_{\sD}(\hat{s}_i) \Big).$$
This gives $$\Big[\pi(\hat{s}_i)\Big(\varphi_1(\hat{s}_i)(\beta^{\sD} - \alpha^{\sD}) - \C_{\sD}(\hat{s}_i) \Big) \Big]( {\bf p'}^1_{\sD}(\hat{s}_i) - {\bf p}^1_{\sD}(\hat{s}_i)) \geqslant 0.$$
The term $\pi(\hat{s}_i)\Big[\varphi_1(\hat{s}_i)(\beta^{\sD} - \alpha^{\sD}) - \C_{\sD}(s_i) \Big]$ is independent of the change in the tagging probability at $\hat{s}_i$ and the value is either positive or negative. By the equilibrium assumption,  $ \pi(\hat{s}_i)\Big[\varphi_1(\hat{s}_i)(\beta^{\sD} - \alpha^{\sD}) - \C_{\sD}(\hat{s}_i) \Big] = 0$ (since $( {\bf p'}^1_{\sD}(\hat{s}_i) - {\bf p}^1_{\sD}(\hat{s}_i) )$ can be made positive or negative and the inequality must hold for both cases). As $\pi(\hat{s}_i) \neq 0$, this implies $$ \varphi_1(\hat{s}_i) =  \frac{\C_{\sD}(\hat{s}_i) }{(\beta^{\sD} - \alpha^{\sD})} .$$ By varying the tag sink selection probability, i.e., changing ${\bf p}_{\sD}^2(\hat{s}_i)$ to ${\bf p'}_{\sD}^2(\hat{s}_i)$, we get $$ \varphi_2(\hat{s}_i) =  \frac{\W_{\sD}(\hat{s}_i) }{(\beta^{\sD} - \alpha^{\sD})}, \mbox{~where~} \varphi_2(\hat{s}_i) :=   \Big( {\bf p}^1_{\sD}(\hat{s}_i) \prod_{\substack{r \in \lambda(\hat{s}_i)}} {\bf p}^{2+r}_{\sD}(\hat{s}_i) \Big) .$$ Similarly, by varying the probability of selection each of the rules at $\hat{s}_i$,  for  $r\in \{1, \ldots, N\}$, $$ \varphi_{r+2}(\hat{s}_i) =  \frac{\gamma_r }{(\beta^{\sD} - \alpha^{\sD})}.$$ Taking logarithms of $\varphi_1(\hat{s}_i), \ldots, \varphi_{2+N}(\hat{s}_i),$ for a node $\hat{s}_i \in \S$, we get $2+N$ independent linear equations with $2+N$ unknowns. Thus there exists a unique solution to this set of equations which indeed gives the defender's policy at Nash equilibrium. Thus solution to the matrix-game in Table~\ref{tb:matrix} gives a NE to the single-stage case.
\end{proof}
This completes the discussion on the NE for the flow tracking game when the attack consists of a single stage. In the next section, we analyze the equilibrium of multi-stage attack.
\section{Results: Solution to Flow Tracking Game for Multi-Stage Attacks}
\label{sec:results2}
Solving for Nash equilibrium in non-zero sum, imperfect information game settings is generally known to be computationally difficult. In this section, we present an efficient algorithm to compute a locally optimal correlated equilibrium~\cite{papadimitriou2008computing,aumann1987correlated} of the game introduced in the Section~\ref{sec:formulation}.

Formal definitions for the correlated equilibrium and the local correlated equilibrium are given in the Definitions~\ref{def:correlated} and~\ref{def:local_correlated}, respectively. Intuitively correlated equilibrium can be viewed as a general distribution over set of strategy profiles such that if an impartial mediator privately recommends action to each player from this distribution, then no player has an incentive to choose a different strategy. The correlated equilibrium has several  advantages~\cite{papadimitriou2008computing,aumann1987correlated} : (1) it is guaranteed to always exist and (2) it can be solved in polynomial time (i.e. computing a correlated equilibrium only requires solving a linear program whereas solving a Nash equilibrium requires finding its fixed point).

In order to find locally optimal correlated equilibrium solutions, we map our two-player game model into a game with $(M+2)N+ \Lambda + 1$ players. Where $\Lambda = \sum_{i=1}^{N}\Lambda(s_i)$ with $s_i \in V_{\G}$ and $\Lambda(s_i)$ stands for the total number of security rules associated with the node $s_i \in V_{\G}$ in the information flow graph. Then the adversary's strategy is represented by $MN + 1$ players, $MN$ of which represents the adversary's actions at every node $s_i \in V_{\G}$ and for a specified stage $j \in \{1,\hdots, M\}$ and one adversarial player acting on the pseudo node, $s_0$,  whose strategy decides the entry point chosen by the adversary into the system. Similarly, the defender's strategy is represented by $\sum_{i=1}^{N}\Lambda(s_i) + 2N$ players, each one of the $\Lambda(s_i) + 2$ defender player represents the defender's strategy (tag selection, trap selection, and selection of  security rules) at a single node $s_i$.

Formally, we consider a set of players $\{\mathcal{P}^{A_{ij}} : i=1,\ldots,N, j=1,\ldots,M\} \cup \{\mathcal{P}^{D_{i}} : i=1,\ldots,2N+ \Lambda\} \cup \{\mathcal{P}^{s_0}\}$. Each of the players in $\mathcal{P}^{A_{ij}}$ has action space $\mathcal{A}^{A_{ij}} = \cN(s_{i})$, each player in $\mathcal{P}^{D_{i}}$ has action space $\{0,1\}$, depending on the type of defender player representing whether or not to tag or trap/no trap or selecting or not selecting a specific tag check rule. The player  $\mathcal{P}^{s_0}$ has action space $\lambda$. We let $\mathbf{a}^{D}$ denote the set of actions chosen by the players $\{\mathcal{P}^{D_{i}} : i=1,\ldots,2N+ \Lambda\} $ and $\mathbf{a}^{A}$ denote the set of actions chosen by the players $\{\mathcal{P}^{A_{ij}} : i=1,\ldots,N, j=1,\ldots,M\} \cup \{  \mathcal{P}^{s_0} \}$.

The payoffs of the players from a particular action set are given by 
\begin{eqnarray*}
	U^{A_{ij}}(\mathbf{a}_{A}, \mathbf{a}_{D}) &=& U^{A}(\mathbf{a}_{A}, \mathbf{a}_{D}), \\
	U^{D_{i}}(\mathbf{a}_{A}, \mathbf{a}_{D}) &=& U^{D}(\mathbf{a}_{A}, \mathbf{a}_{D}),
\end{eqnarray*}
where $U^{A}$ and $U^{D}$ are as defined in Section~\ref{sec:formulation}. Hence, all adversarial players receive the same utility $U^{A}$, while all defender players receive the utility $U^{D}$. Equivalently, the adversarial players $U^{A_{ij}}$ cooperate in order to maximize the adversary's utility, while the defender players $U^{D_{i}}$ attempt to maximize the defender utility. 

Under the solution algorithm, the game is played repeatedly, with each player choosing its action from a probability distribution (mixed strategy) over the set of possible actions. After observing their utilities, the players update their strategies according to an \textit{internal regret minimization} learning algorithm \cite{CesaLug-06}. A pseudo-code of the proposed algorithm for computing correlated equilibrium strategies for both defender and adversary players is given in Algorithm \ref{algo:correlated}.
\begin{algorithm}[h]
	\caption{Pseudo-code of the algorithm for computing  correlated equilibrium\label{algo:correlated}}
	\begin{algorithmic}[1]
		\State Initialize $t \leftarrow 0$
		\For{$n=1,\ldots,(M+2)N+ \Lambda + 1$}
		\State $\mathbf{p}_{t,n} \leftarrow$ uniform distribution over set of actions
		\EndFor
		\While{$||\mathbf{p}_{t} - \mathbf{p}_{t-1}|| > \epsilon$}
		\For{$n=1,\ldots,(M+2)N+ \Lambda + 1$}
		\State $a_{t,n} \leftarrow $ action chosen from distribution $\mathbf{p}_{t,n}$
		\EndFor
		\For{$n=1,\ldots,(M+2)N+ \Lambda + 1$}
		\State $a_{t,-n} \leftarrow (a_{t,l} : l \neq n)$
		\For{all $(r,s)$ actions of player $n$}
		\State $\mathbf{p}_{t,n}^{r \rightarrow s} \leftarrow \mathbf{p}_{t,n}$
		\State $\mathbf{p}_{t,n}^{r \rightarrow s}(r) \leftarrow 0$
		\State $\mathbf{p}_{t,n}^{r \rightarrow s}(s) \leftarrow \mathbf{p}_{t,n}(r) + \mathbf{p}_{t,n}(s)$
		\State $\Delta_{(r,s),t,n} \leftarrow \frac{\exp{\left(\eta\sum_{u=1}^{t-1}{\mathbf{E}(U^{n}(\mathbf{p}_{u,n}^{r \rightarrow s}, a_{u,-n}))}\right)}}{\sum_{(x,y): x \neq y}{\exp{\left(\eta\sum_{u=1}^{t-1}{{\mathbf E}(U^{n}(\mathbf{p}_{u,n}^{x \rightarrow y}, a_{u,-n}))}\right)}}}$
		\State $\mathbf{p}_{t,n} \leftarrow$ fixed point of equation $\mathbf{p}_{t,n} = \sum_{(i,j): i \neq j}{\mathbf{p}_{t,n}^{r \rightarrow s}\Delta_{(r,s),t,n}}$
		\EndFor
		\EndFor
		\State $t \leftarrow t+1$
		\EndWhile
	\end{algorithmic}
\end{algorithm}

The algorithm initializes the strategies at each node of the information flow graph to be uniformly random. At each iteration $t$, an action  is chosen for each player according to the probability distribution $\mathbf{p}_{t,n}$ of player $n$. After observing the actions from other players, the probability distribution $\mathbf{p}_{t,n}$ is updated as follows. For each pair of actions $r$ and $s$, the  new probability distribution $\mathbf{p}_{t,n}^{r \rightarrow s}$ is generated, in which all of the probability mass allocated to action $r$ is instead allocated to action $s$. The expected utility arising from $\mathbf{p}_{t,n}^{r \rightarrow s}$ can be interpreted as the expected benefit from playing action $s$ instead of $r$ at previous iterations of the algorithm. 

For each pair $(r,s)$, a weight $\Delta_{(r,s),t,n}$ is computed that consists of the relative benefit of each distribution $\mathbf{p}_{t,n}^{r \rightarrow s}$, i.e., pairs $(r,s)$ such that allocating probability mass from $r$ to $s$ produces a larger expected utility will receive higher weight. A new distribution $\mathbf{p}_{t,n}$ is then computed based on the weights $\Delta_{(r,s),t,n}$, so that actions that produced a higher utility for the player will be chosen with increased probability. The algorithm continues until the distributions converge.  The convergence of the algorithm is described by the following proposition.

\begin{prop}
	\label{prop:algo-convergence}
	Algorithm \ref{algo:correlated} converges to a local correlated equilibrium of the game introduced in Section~\ref{sec:formulation}.
\end{prop}

\begin{proof}
	By \cite{CesaLug-06}, Algorithm~\ref{algo:correlated} converges to a correlated equilibrium of the $((M+2)N+ \Lambda + 1)$-player game. Equivalently, by Definition~\ref{def:correlated}, for any $s_{i}$ and $p_{D_{i}}^{\prime} \in [0,1]$,  the joint distribution $P$ returned by the algorithm satisfies 
	\begin{equation}
		\label{eq:correlated-1}
		\mathbf{E}(U^{D_{i}}(p_{D_{i}}, \mathbf{p}_{D_{-i}}, \mathbf{p}_{A})) \geq \mathbf{E}(U^{D_{i}}(p_{D_{i}}^{\prime}, \mathbf{p}_{D_{-i}}, \mathbf{p}_{A})).
	\end{equation}
	Since the utility $U^{D_{i}}$ is equal to $U^{D}$ for all $i \in \{1,\ldots,N\}$, Eq. (\ref{eq:correlated-1}) is equivalent to 
	\begin{equation}
		\label{eq:correlated-2}
		\mathbf{E}(U^{D}(p_{D_{i}}, \mathbf{p}_{D_{-i}}, \mathbf{p}_{A})) \geq \mathbf{E}(U^{D}(p_{D_{i}}^{\prime}, \mathbf{p}_{D_{-i}}, \mathbf{p}_{A})).
	\end{equation}
	Similarly, for any $s_{i}^j \in \{\S\times \{1,\ldots,M\}\}\cup\{ s_0^1 \}$ and any $p_{A_{ij}}^{\prime}$, we have 
	\begin{equation}
		\label{eq:correlated-3}
		\mathbf{E}(U^{A_{ij}}(\mathbf{p}_{D}, \mathbf{p}_{A_{-ij}},p_{A_{ij}})) \geq \mathbf{E}(U^{A_{ij}}(\mathbf{p}_{D}, \mathbf{p}_{A_{-ij}},p_{A_{ij}}^{\prime}))
	\end{equation}
	which is equivalent to 
	\begin{equation}
		\label{eq:correlated-4}
		\mathbf{E}(U^{A}(\mathbf{p}_{D}, \mathbf{p}_{A_{-ij}},p_{A_{ij}})) \geq \mathbf{E}(U^{A}(\mathbf{p}_{D}, \mathbf{p}_{A_{-ij}},p_{A_{ij}}^{\prime}))
	\end{equation}
	Equations (\ref{eq:correlated-2}) and (\ref{eq:correlated-4}) imply that the output of Algorithm \ref{algo:correlated} satisfies the conditions of Definition \ref{def:local_correlated} and hence is a local correlated equilibrium.
\end{proof}

\noindent The following Proposition provides the complexity analysis for the proposed algorithm.

\begin{prop}
	\label{prop:complexity}
	With probability $(1-\zeta)$, Algorithm \ref{algo:correlated} returns an $\epsilon$-correlated equilibrium using $O\left(\frac{N^{2}(M+2)+N(\Lambda + 1)}{\epsilon^{2}}\ln{\left(\frac{N^{2}(M+1)+N(\Lambda + 1)}{\zeta}\right)}\right)$ evaluations of the utility function.
\end{prop}

\begin{proof}
	By \cite[Chapter~7, Section~7.4]{CesaLug-06},  learning-based algorithms return an $\epsilon$-correlated equilibrium with probability $(1-\zeta)$ within $\max_{n}{\frac{16}{\epsilon^{2}}\ln{\frac{N_{n}K}{\zeta}}}$ iterations, where $N_{n}$ is the number of actions for player $n$ and $K$ is the number of players, incurring a total of $\frac{16N_{n}K}{\epsilon^{2}}\ln{\frac{N_{n}K}{\zeta}}$ evaluations of the utility function. In this case, $N_{n} \leq N$ and $K=(M+2)N+ \Lambda + 1$, resulting in the desired complexity bounds. 
\end{proof}
Proposition \ref{prop:complexity} shows that convergence of the algorithm is sublinear in the number of nodes, with a total complexity that is quadratic in the number of nodes and linear in the number of stages.
\section{Experimental Analysis}\label{sec:sim}
In this section, we  provide the experimental validation of our model and  results using  real-world attack data set obtained using  Refinable Attack INvestigation system (RAIN) \cite{JiLeeDowWanFazKimOrsLee-17}, \cite{JiLeeFazAllDowKimOrsLee-18}  for a three day nation state attack.  We implement our model and run Algorithm~\ref{algo:correlated} on the information flow graph generated using the system log data obtained using the RAIN system for the day one of the nation state attack.  Using the results obtained we verify the correctness of the proposed algorithm and also perform sensitivity analysis by varying the cost of defense (i.e., tagging costs, trapping costs and individual security rule selection costs at each nodes in the underlying information flow graph) for the defender. 
This analysis enable us to infer the optimal strategies of the players and the sensitivity of the model with respect to cost parameters for a given attack data set (information flow graph with specified destinations for each attack stage). In order to apply the proposed analysis on any real-time data attack data set, we first construct the information flow graph for the system under consideration and then run Algorithm~\ref{algo:correlated} on this graph to obtain the defender's policy, i.e., tagging locations, trapping locations and selection of appropriate security rules,  at a local equilibrium of the multi-stage game.
\begin{figure*}[t]
	\centering
	$\begin{array}{cc}
	\includegraphics[width=3in]{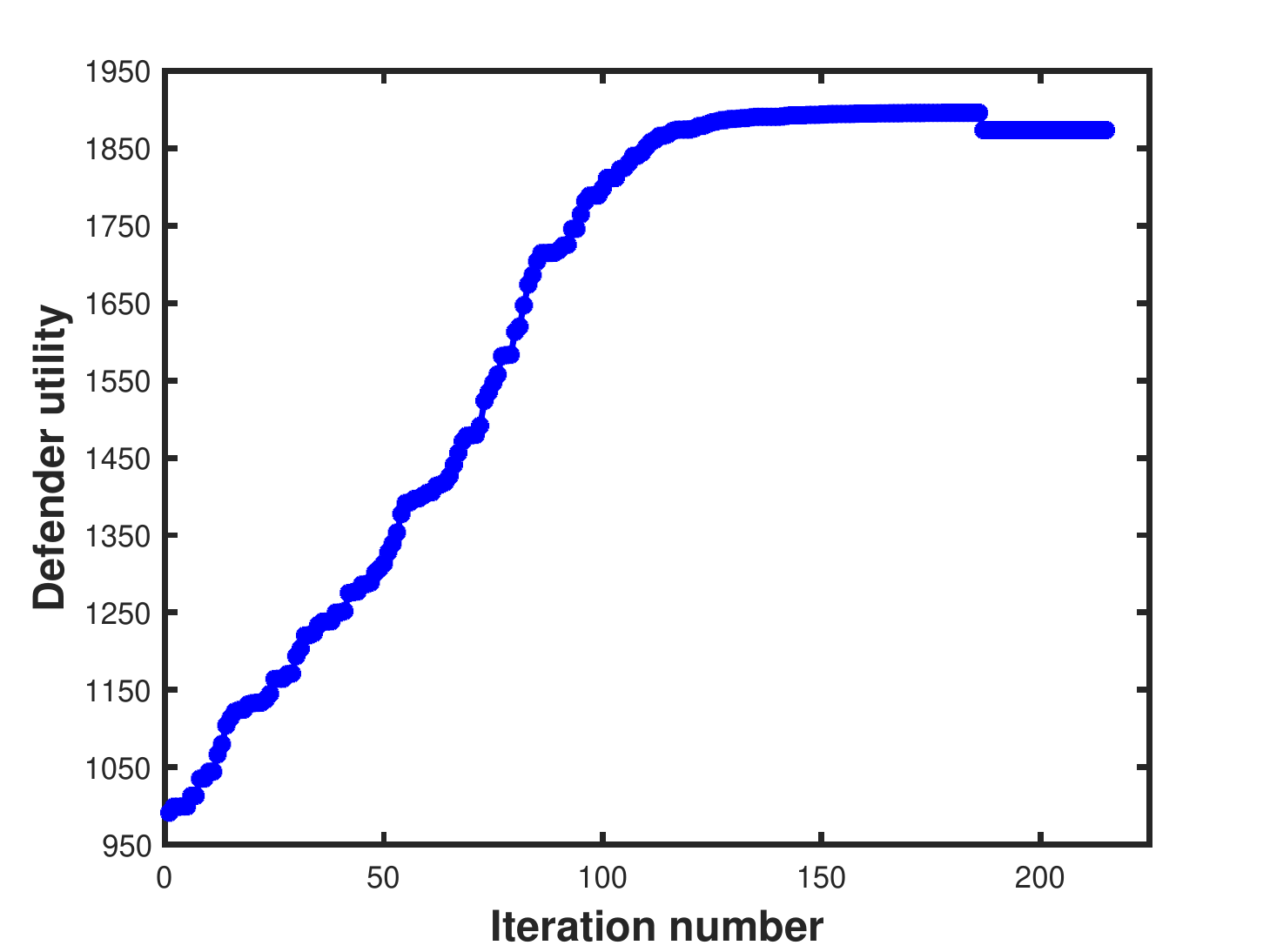} &
	\includegraphics[width=3in]{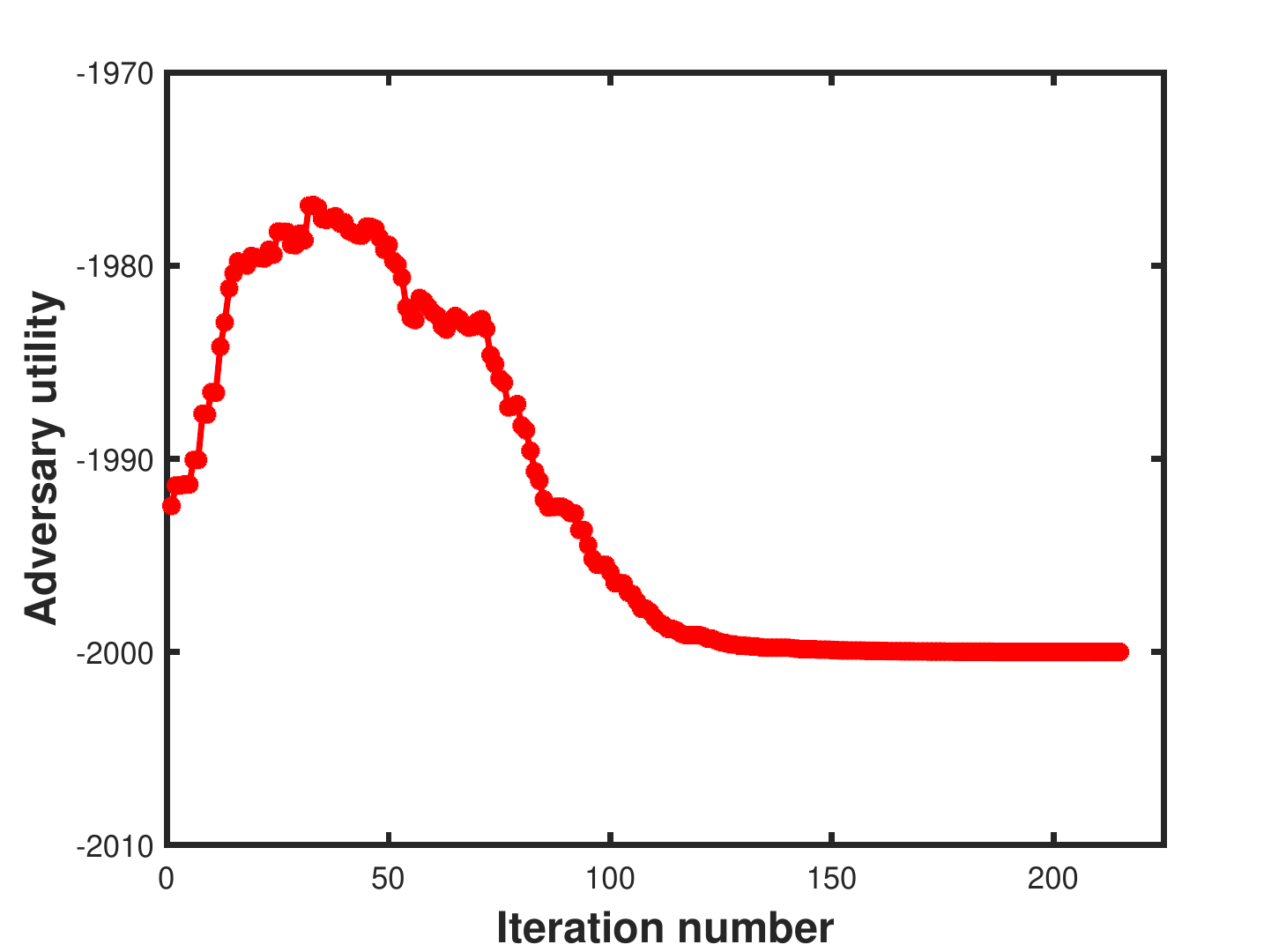}\\
	\mbox{(a)} & \mbox{(b)}
	\end{array}$
	\caption{(a) Average utility of the defender and (b) Average utility of the adversary, at each iteration of  Algorithm~\ref{algo:correlated} with cost parameters of the game  set as follows: $\beta^A_1 = 100, \beta^A_2 = 200, \beta^A_3 = 500, \beta^A_4 = 1200,$ $\alpha^A = -2000, \alpha^D  = 2000$, $\beta^D_1 = -100, \beta^D_2 = -200, \beta^D_3 = -500, \beta^D_4 = -1200$. The costs of tagging and trapping  at each node are set with fixed cost value of $c_1=-50$ and $c_2 =-50$, respectively, multiplied by the fraction of flows through each node in the information graph. The cost of selecting the security rules are set as $\gamma_1 = \ldots = \gamma_N =-50$. For simulation purposes we assume that the fraction of the flows through each node in the information flow graph is uniform.} \label{fig:1}
\end{figure*}

We present below  the details of the attack we consider and the steps involved in the construction of the information flow graph for that attack.

\subsection{Attack Description}
The evaluation was completed on a nation state attack (i.e., state-of-the-art APT (Advanced Persistent Threat) attack) orchestrated by a red-team during an adversarial engagement. The engagement was organized by a US government agency  (US DARPA). During the engagement, we leveraged RAIN~\cite{JiLeeDowWanFazKimOrsLee-17} to record the whole-system  log. At a high-level the goal of the adversaries' campaign was to steal sensitive proprietary and personal information from the targeted company. The attack is designed to run through three days. We only consider the day one log data collected via RAIN for our evaluation purposes. Through our extensive analysis,  we partitioned the attack in day 01 into four key stages: initial compromise, internal reconnaissance,  foothold establishment, and data exfiltration. The initial compromise leveraged a spear-phishing attack,  which lead the victim to a website that was hosting ads from a malicious web server. The victim navigated to the website, which exploited a vulnerability in the Firefox browser. Once the attackers had compromised the machine, the next stage of the APT leveraged common utilities to do internal reconnaissance. The goal of this stage was to fingerprint the compromised system to detect running processes and network information. Next, the attackers established a foothold by writing a malicious program to disk. The malicious program was eventually executed, and established a backdoor, which was used to  continuously exfiltrate the companies sensitive data. 

The system log data for the day 01 of the nation state attack is collected with the annotated entry points of the attack and the attack destinations corresponding to each stage of the attack. Initial conversion of the attack  data into an information flow graph resulted in a coarse-grain graph with $\approx$ 1,32,000 nodes and $\approx$ 2 million edges. Coarse graph captures the whole system data during the recording time which includes the attack related data and lots of data related to system's background processes (noise). Hence coarse graph provides very little security sensitive (attack related)  information about the underlying system and it is computationally intensive to run our algorithm on such coarse graph. Without loss of any relevant  information, now we prune the coarse graph to extract the security sensitive information about the system from the data \cite{JiLeeDowWanFazKimOrsLee-17}.  The resulting  information flow graph is called as {\em security sensitive information sub-graph} \cite{JiLeeDowWanFazKimOrsLee-17} and we run our experimental analysis on this graph. The pruning includes the following two major steps:
\begin{enumerate}
	\item Starting from the provided attack destinations, perform up stream, down stream and point to point stream techniques discussed in \cite{JiLeeDowWanFazKimOrsLee-17} to prune the coarse information flow graph \label{step1}
	\item Further prune the resulting subgraph from Step~\ref{step1}) by combining object nodes (e.g. files, net-flow objects) that belong to the same directories or that  use the same network sockets. 
\end{enumerate}

\begin{figure}[h]
	\centering
	{\includegraphics[width=0.5\textwidth]{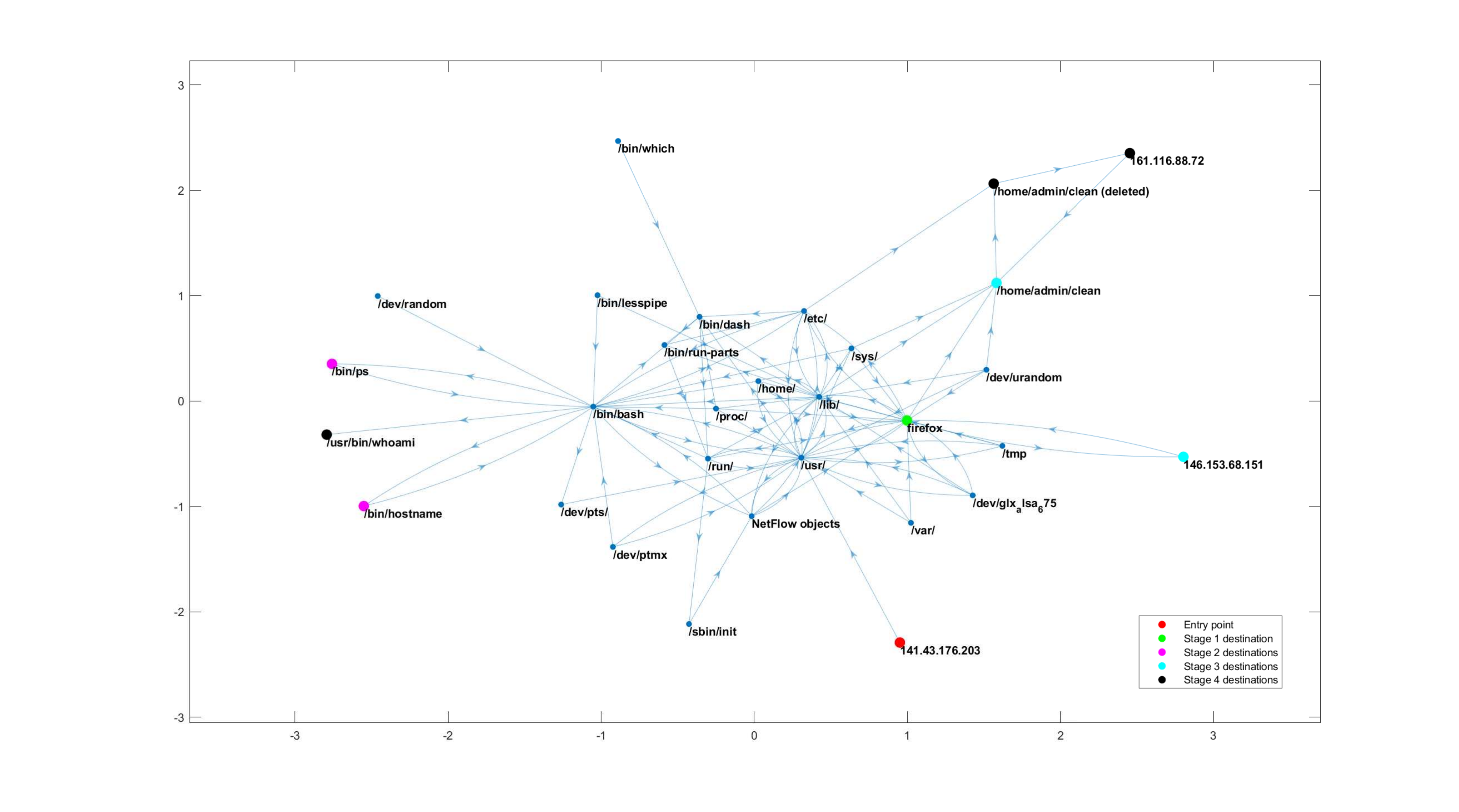}}
	\hspace{\parindent}
	\caption{Security sensitive information flow sub-graph for the nation state attack}\label{fig:3}
\end{figure}

The resulting pruned information flow graph consists of $30$ nodes ($N=30$) out of which $8$ nodes are identified as attack destination nodes corresponding to each of the $4$ stages ($M=4$) of the day 01 nation state attack. One node related to a net-flow object has been identified as an entry point used for the attack ($|\lambda|=1$). Note that, even when the sensitive locations in the system are known it may not be feasible to do tagging, trapping, and security analysis at that location (entry point of attack) which is captured in our model by the costs for tagging, trapping, and performing authenticity analysis using a security rule.
\subsection{Case Study 1: Convergence of the algorithm}\label{study1}
In this section, we provide a case study that validates the convergence of the proposed algorithm. The game parameters (rewards, penalties and costs)  of the players are set as follows: $\beta^A_1 = 100, \beta^A_2 = 200, \beta^A_3 = 500, \beta^A_4 = 1200$, $\alpha^A = -2000$, $\alpha^D  = 2000$, $\beta^D_1 = -100, \beta^D_2 = -200, \beta^D_3 = -500, \beta^D_4 = -1200$. We set costs of tagging and  trapping at each node to have fixed cost value of $c_1=-50, c_2 = -50$, respectively, multiplied by the fraction of flows through each node in the information graph and the costs for selecting security rules as $\gamma_1=\ldots=\gamma_N=-50$. For simulation purposes we assume fraction of the flows through each node is uniformly distributed. One can estimate the distribution for the fraction of flows through each node by counting number of events associated with each node in the RAIN whole system provenance and normalizing those values with the total number of events occurred during the recording. Figure~\ref{fig:1} plots the utility values for both players at each iteration of Algorithm~\ref{algo:correlated} for the above-mentioned parameters. It shows that both defender and adversary utilities converges within finite number of iterations.

\subsection{Case Study 2: Utility of the defender vs. defense cost}
This case study is used to analyze the effect of the cost of defense on the defender's utility. We use same game parameters as used in the Case Study~\ref{study1}. Then at each experiment round we scale all three defense cost components of the defender by scaling factors 0.01, 0.1, 0.5, 1, 3, 6 and 10. Figure~\ref{fig:2} shows that the expected utility of the defender starts decreasing exponentially when the cost of defense increases to higher values. When the defense cost increases defender incurs more resource cost to maintain the same level of security (maintain low attack success probability by setting up tag sources and traps) in the system or higher costs keeps defender away from frequently deploying tag sources and traps in the system and hence attack successes rate increases and defender's utility decreases. This results hence implies that an optimal selection of the locations for tagging and trapping is critical for a resource-efficient implementation of detection system against APTs.
\begin{figure}[h]
	\centering
	{\includegraphics[width=0.4\textwidth]{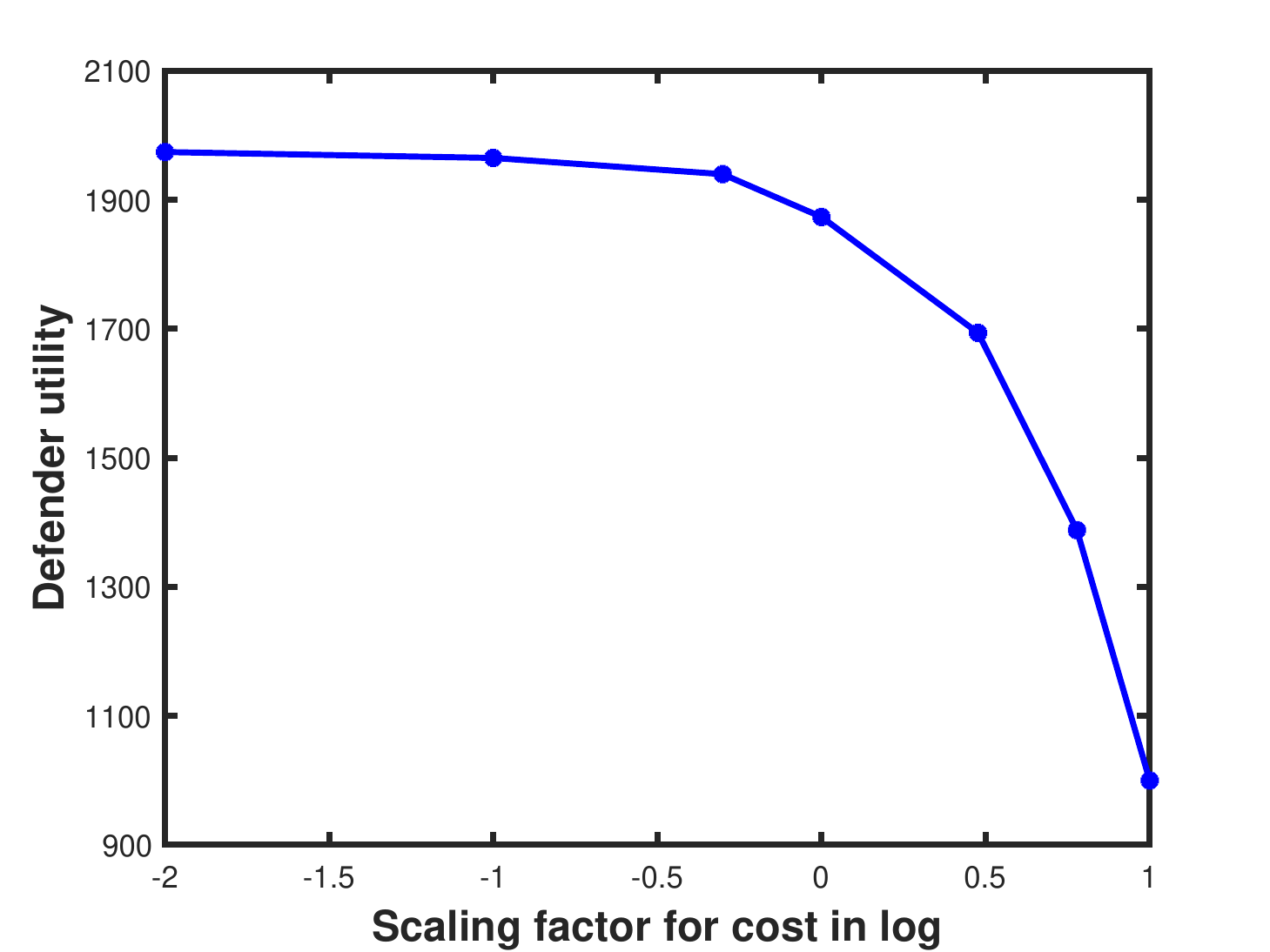}}
	\hspace{\parindent}
	\caption{Average utility of defender as a function of the cost for defense. In each realization of the experiment we scale the components of the cost for defense (i.e. costs for tag, trap and tag check rule selection) by a increasing scaler factors 0.01, 0.1, 0.5, 1, 3, 6 and 10. All the other game parameters have been fixed to constant values used in the Case Study~\ref{study1}.}\label{fig:2}
\end{figure}
\section{Conclusions}\label{sec:conclu}
This paper proposed a game theoretic framework for cost-effective real-time detection of Advanced Persistent Threats (APTs) that perform a multi-stage attack. We used an information flow tracking-based detection mechanism as APTs continuously introduce information flows in the system. As the defense mechanism is unaware of the stage of the attack and also can not distinguish whether an information flow is malicious or not, the game considered in this paper has asymmetric information structure.  Further, the resource costs of the attacker and the defender are not same resulting in a multi-stage nonzero-sum imperfect information game. We first computed the best responses of both the players. For the adversary, we showed that the best response can be obtained using a shortest path algorithm in polynomial time. For the defender, we showed that the utility function of the defender is submodular  for a given adversary strategy and hence a $1/2$-optimal\footnote{An $\epsilon$-optimal solution is a feasible solution whose value is at most $\epsilon$ times that of the actual optimum value.} solution to the best response can be found in polynomial time. For solving the game in order to obtain an optimal policy for the defender, we first considered the single-stage attack case and characterized the set of Nash equilibria by proving the equivalence of the sequential game to a bimatrix-game formulation. Then, we considered the multi-stage case and provided a polynomial time algorithm to find an $\epsilon$-correlated equilibrium, for any given $\epsilon$. We performed experimental analysis of our model on real-data for a three day nation state attack obtained using Refinable Attack INvestigation (RAIN) system.
\bibliographystyle{myIEEEtran}      
\bibliography{MURI_references}
\end{document}